%% file: main.tex
\documentclass[11pt]{article}

\usepackage{amssymb,amsthm,amsmath}

\usepackage{enumitem}
\setlist[enumerate]{topsep=0pt}

\usepackage{tikz}
\usepackage[font=scriptsize,justification=centering]{caption}
\usepackage{thm-restate}
\usetikzlibrary{decorations.text}
\usepackage[ruled]{algorithm2e}

\definecolor{DarkRed}{rgb}{0.5,0.1,0.1}
\definecolor{DarkBlue}{rgb}{0.1,0.1,0.5}
\definecolor{ForestGreen}{rgb}{0.1333,0.5451,0.1333}
\definecolor{Red}{rgb}{0.9,0,0}
\usepackage[linktocpage=true, pagebackref=true,colorlinks,urlcolor=red,linkcolor=DarkRed,citecolor=ForestGreen,bookmarks,bookmarksopen,bookmarksnumbered]{hyperref}

\usepackage{thmtools,thm-restate}
\usepackage{comment}
\usepackage{cleveref}

\crefname{property}{property}{Property}
\creflabelformat{property}{(#1)#2#3}
\crefname{equation}{eq}{Eq}
\creflabelformat{equation}{(#1)#2#3}
\usepackage{fullpage}
\usepackage{bbm}
\usepackage{soul}
\usepackage{thm-restate}

\newtheorem{theorem}{Theorem}[section]

\newtheorem{lemma}[theorem]{Lemma}

\newtheorem{remark}[theorem]{Remark}
\newtheorem{claim}[theorem]{Claim}
\newtheorem{corollary}[theorem]{Corollary}

\newtheorem{fact}[theorem]{Fact}

\crefname{algorithm}{algorithm}{algorithms}
\Crefname{algorithm}{Algorithm}{Algorithms}
\usepackage[most]{tcolorbox}

\theoremstyle{definition}
\newtheorem{definition}{Definition}

\input{macros}

\title{Efficient Algorithms for Online Subadditive\\ Combinatorial Allocations}
\author{
Calum MacRury \thanks{Georgia Tech, School of Industrial and Systems Engineering (ISyE), 
\texttt{calum.macrury@isye.gatech.edu}} \and Rian Neogi \thanks{University of Waterloo, \texttt{rneogi@uwaterloo.ca}} \and Kanstantsin Pashkovich \thanks{University of Waterloo, \texttt{kpashkovich@uwaterloo.ca}} \and \and
Sahil Singla\thanks{School of Computer Science, Georgia Tech. \texttt{ssingla@gatech.edu}. Supported in part by NSF awards CCF-2327010 and CCF-2440113.}
\and
Siddharth M. Sundaram\thanks{School of Computer Science, Georgia Tech. \texttt{ssundaram38@gatech.edu}. Supported in part by NSF awards CCF-2327010 and CCF-2440113.} \and Chaitanya Swamy \thanks{University of Waterloo, \texttt{cswamy@uwaterloo.ca}}
}
\date{}

\begin{document}

\maketitle


\begin{abstract}
    \smallskip

    For the online combinatorial allocation problem with subadditive valuations, Correa and Cristi (STOC 2023) proved the existence of a $6$-competitive online algorithm, improving on the previous best $O(\log\!\log m)$-competitive online algorithm due to D\"utting, Kesselheim, and Lucier (FOCS 2020), where $m$ is the number of items. However, Correa and Cristi's result is existential, and it was left open whether a constant competitive ratio is attainable via an \textit{efficient} online algorithm that uses a polynomial number of demand oracle queries. In this work, we answer this affirmatively,
    giving an expected-polynomial-time $(6 + \epsilon)$-competitive online algorithm for any constant $\epsilon > 0$. Our techniques also recover, in the offline setting, the $(2 + \epsilon)$-approximation result of Feige (STOC, 2006). Finally, when the buyers' valuations are drawn
    from identical distributions,  we exploit symmetry to obtain an improved $(60/11+\epsilon)$-competitive algorithm.

    \smallskip

    Starting with the natural configuration LP relaxation for the problem, 
    our main technical contribution is a recursive \emph{Bundle Score Generator} (BSG) that resolves item conflicts by assigning correlated scores to the items requested by each buyer.
    Unlike the Random Score Generator whose existence is shown by Correa and Cristi, our BSG is efficiently computable in expected polynomial time. Moreover, it satisfies a stochastic dominance property that is sufficient to recover both Feige's offline result and Correa and Cristi's online result. 



\end{abstract}



\vspace{0.5cm}
\newpage

{\small
\setcounter{tocdepth}{1}
\tableofcontents
}

\thispagestyle{empty}
\setcounter{page}{0}

\newpage


\input{intro}

\input{preliminaries}

\input{recursive_procedure}

\input{offline_rounding}

\input{GW_RSG}

\input{iid}



\subsection*{AI Declaration}

The authors used GPT-5.4 Pro for assistance in obtaining the construction of the $\Height$ procedure of \Cref{sec:bsg}, its value retention guarantee and the subsequent improvement for the i.i.d. setting. 
AI tools were also used for assistance with writing. The authors have checked the arguments and take full responsibility for all content of the paper.


\appendix
\input{fixed_point}

\input{MMS-appendix.tex}
\input{XOS-Appendix}

\input{Feige-appendix.tex}
\input{Odd-n-Appendix}

{\small
\bibliographystyle{alpha} 
\bibliography{refs}
}

\end{document}

%% file: macros.tex
\newcommand{\sid}[1]{\textcolor{red}{[\textbf{Sid}: #1]}}
\newcommand{\sahil}[1]{\textcolor{blue}{[\textbf{Sahil}: #1]}}
\newcommand{\CMc}[1]{\textcolor{orange}{[\textbf{Calum}: #1]}}

\newcommand{\separator}[1]{\textcolor{cyan}{[#1]}}

\usepackage{color-edits}
\addauthor{Cal}{orange}

\newcommand{\ALG}{\mathrm{ALG}}
\newcommand{\OPT}{\mathrm{OPT}}

\newcommand{\poly}{\ensuremath{\mathsf{poly}}}
\newcommand{\E}{\mathbb{E}}

\newcommand{\R}{\mathbb{R}}

\newcommand{\ignore}[1]{{}}

\newcommand{\mb}{\mathbb}

\newtheorem{mdalg}{Algorithm}

\newenvironment{tbox}{\begin{tcolorbox}[
		enlarge top by=5pt,
		enlarge bottom by=5pt,
		breakable,
		boxsep=0pt,
		left=4pt,
		right=4pt,
		top=10pt,
		arc=0pt,
		boxrule=1pt,toprule=1pt,
		colback=white
		]
	}
	{\end{tcolorbox}}

\newcommand{\Z}{\mathbb{Z}}
\newcommand{\Height}{\ensuremath{\mathsf{RecHeight}}}
\newcommand{\bid}{score}

\newcommand{\RSG}{\mathcal{S}}
\newcommand{\VD}{\mathcal{D}}
\newcommand{\supp}{\text{supp}}
\newcommand{\typeprob}{q}
\newcommand{\type}{\tau}
\newcommand{\typeset}{\Gamma}
\newcommand{\shsg}{\text{bundle score generator}}

\newcommand{\SHSG}{\text{BSG}}
\newcommand{\SHSGs}{\text{BSGs}}

\newcommand{\Ber}{\textup{Ber}}
\newcommand{\OCRS}{\mathsf{OCRS}}

%% file: intro.tex
\section{Introduction}

In online combinatorial allocation, a seller has a fixed supply of $m$ non-identical items, and $n$ buyers arrive sequentially. When the $i$'th buyer arrives, their valuation function $v_i$ over bundles of items is revealed, and the online algorithm must irrevocably allocate a subset of the remaining items to them. The goal is to maximize the total value, or welfare, of the final allocation.

For worst-case inputs, no $o(n)$-competitive algorithm is possible, since the single-item problem already captures the difficulty of competing against the maximum of an online sequence of numbers.
 This has motivated a long line of work on the stochastic/prophet inequality version of the problem, where each buyer's valuation is drawn independently from a known prior distribution.  The seminal work of Feldman, Gravin, and Lucier~\cite{feldman2014combinatorial} attained a \(2\)-competitive algorithm for fractionally subadditive valuations, and this factor is best possible even for a single item~\cite{krengel1977semiamarts,krengel1978semiamarts}.
  We note that fractionally subadditive valuations are equivalent to XOS valuations \cite{feige2006maximizing}, and the class strictly contains submodular valuations.


For the broader class of \emph{subadditive valuations}, the picture is much less clear. A valuation \(v:2^{[m]}\to\mathbb R_{\ge0}\) is subadditive if
\[
    v(S\cup T)\le v(S)+v(T)
    \qquad \forall S,T\subseteq[m].
\]
Subadditive valuations are the standard formalization of complement-free preferences in economics, and they strictly contain XOS valuations. 
Since every subadditive valuation can be \(O(\log m)\)-approximated by an XOS valuation~\cite{Dobzinski-Approx07,BhawalkarR-SODA11}, an \(O(\log m)\)-competitive algorithm immediately follows from the \(2\)-competitive algorithm for XOS valuations. Zhang~\cite{zhang2022improved} improved this factor to \(O(\log m/\log\!\log m)\),
and later D\"utting, Kesselheim, and Lucier~\cite{dutting2020log} obtained an \(O(\log\!\log m)\)-competitive guarantee.


A breakthrough result of Correa and Cristi~\cite{correa2023constant} finally
established a constant-factor guarantee for subadditive valuations, specifically a \(6\)-competitive algorithm. 
However, their result is not known to be efficiently implementable: it relies on the Kakutani fixed-point theorem to show the existence of a suitable random score generator, and the authors explicitly note that they do not know how to construct such scores in polynomial time.


This algorithmic gap is especially striking in light of the offline problem: for subadditive welfare maximization with known valuations,
Feige~\cite{feige2006maximizing} gave a polynomial-time \((2+\epsilon)\)-approximation using demand-oracle queries, for arbitrarily small constant \(\epsilon>0\).  Thus, the best known polynomial-time guarantee for the online subadditive problem was the $O(\log\!\log m)$ competitive guarantee of D\"utting, Kesselheim, and Lucier \cite{dutting2020log}, leaving open the following algorithmic question:
\begin{quote}
    \emph{Is there a polynomial-time constant-factor competitive algorithm for the online combinatorial allocation problem with subadditive valuations?}
\end{quote}

\paragraph{Our results.} We answer the above question in the affirmative, and in fact match the competitive guarantee of Correa and Cristi~\cite{correa2023constant}, up to an arbitrarily small additive loss:

\begin{theorem} 
\label{thm: non-identical}
For every $\epsilon>0$, there is an online algorithm for the combinatorial allocation problem with subadditive valuations whose expected welfare is at least a $1/(6+\epsilon)$ fraction of the expected offline optimum. The algorithm has expected running time polynomial in $n$, $m$, $1/\epsilon$, and the total support size of the valuation distributions, given demand-oracle access. 
\end{theorem} 

\emph{Remark.} Here, the valuation distributions are assumed to have finite support; in
\Cref{sec: infinite-type}, we show how to remove this assumption given
sampling-oracle access to the distributions and make the running time guarantee independent of the support sizes of the valuation distributions.
Also, we note that by standard capped termination arguments, the expected running time guarantee of the algorithm can be made worst-case polytime.

Our main technical contribution is a
\emph{bundle score generator (BSG)}, which is a score-based rule for rounding a solution to the problem's standard configuration linear program (LP). As we explain in \Cref{sec:techniques},  the BSG is designed to recreate the key ``symmetric covering'' property behind Feige's offline rounding algorithm in a way that can be adapted to the online setting. 
The high-level idea is that once a tentative bundle for a buyer is sampled from the LP solution, the buyer should retain a subset whose value is a constant fraction of the entire bundle's value. 
We determine these assignments via buyer-dependent random scores on the items: each item is assigned to the buyer who gives it the highest score.
This BSG approach yields an algorithm for the offline problem with the same approximation ratio as in \cite{feige2006maximizing}.
Unlike \cite{feige2006maximizing},
our approach naturally extends to the online setting, where we use simulated scores for the competition from buyers who have yet to arrive. 

Finally, when the buyers' valuation distributions are \emph{identical}, the additional symmetry allows us to sharpen the analysis of the same algorithm and obtain a better competitive ratio:

\begin{theorem}
\label{thm: identical}
    Suppose the buyers' valuation distributions are identical. For every
    $\epsilon>0$, there is an online algorithm for subadditive combinatorial
    allocation whose expected welfare is at least
    $1/(60/11+\epsilon)$ fraction of the expected offline optimum. The algorithm has the same oracle access and running-time guarantees as in \Cref{thm: non-identical}.
\end{theorem}


\input{our_techniques}

\input{related_work}

%% file: our_techniques.tex
\subsection{Our Techniques} \label{sec:techniques}

In order to explain our techniques, we first consider the offline problem and the approach taken in \cite{feige2006maximizing}. Suppose we are given
an optimal solution to the standard configuration LP, where LP variable $x_{i,A}$ represents the probability that a bundle $A$ is given to buyer $i$. Given such a fractional solution, we want to round it into an integral allocation, while retaining a constant fraction of the solution's value.
A natural approach is to first have each buyer $i$ independently draw a tentative bundle $T_i$ 
via the distribution over subsets given by $(x_{i,A})_{A \subseteq [m]}$. 
We must then resolve conflicts among items that are tentatively requested by multiple buyers. For XOS
valuations, it is enough to do this on an item by item basis: 
for each buyer \(i\), it suffices that each item in \(T_i\) is retained with constant probability. This can be done by having each item execute its own \textit{offline contention resolution scheme} \cite{feige2006maximizing},
and the same approach extends to the online version
of the problem via \textit{online contention resolution} (see \Cref{app:xos-ocrs-reduction}). 

\paragraph{Prior approaches and their limitations.}
On the other hand, for subadditive valuations, marginal item guarantees do \textit{not} suffice to get a constant approximation, even for the offline problem \cite{feige2006maximizing}. 
 Feige's rounding algorithm instead succeeds by achieving a much stronger
``symmetric covering'' property. Suppose buyer $i$ has tentative bundle $T_i = A$,
and $R\subseteq A$ is the final subset they are assigned/retain. Then, a symmetric covering
ensures that there exists identically distributed $R_1, R_2 \subseteq A$,
whose union is $A$, and for which $R_1 \subseteq R$.

A symmetric covering immediately implies the required value retention for their valuation $v_i$. First, $\E\bigl[v_i(R_1)+v_i(R_2)\bigr] = 2\E[v_i(R_1)]$. Moreover, monotonicity 
yields $v_i(R) \ge v_i(R_1)$, and subadditivity
implies that $v_i(R_1)+v_i(R_2) \ge v_i(A)$. By putting
everything together and taking expectations,
\begin{equation} \label{eqn:strong_self_comp}
    \E[v_i(R)] \ge \E[v_i(R_1)]
    ~=~
    \frac12\E\bigl[v_i(R_1)+v_i(R_2)\bigr]
    ~\ge~
    \frac12 v_i(A),
\end{equation}

After averaging over all tentative bundles and buyers, this gives a $2$-approximation to the LP value. 
In the offline setting, Feige gave a clever guiding-graph rounding approach that achieves this symmetric covering property, and hence a $(2+\epsilon)$-approximation for subadditive valuations.
The challenge in extending this approach to the online setting is that Feige's rounding requires simultaneous access to all buyers' tentative bundles. In the online setting, the tentative bundles of the future buyers are unavailable to the online algorithm; e.g., already with Bernoulli distributions, they depend on which subset of the future buyers have non-zero valuations.

Correa and Cristi~\cite{correa2023constant} take a different route.
They introduce a \textit{random score generator (RSG)},
which is a randomized map from a buyer's valuation to a real-valued score vector on the items. Using RSGs, they first draw base threshold scores, indicating the required scores to attain the various items. Specifically, when each buyer arrives, it generates its own random scores on the items via an RSG, and is then assigned all remaining items on which its score exceeds the corresponding threshold. The drawback of this algorithm is that the RSG used in their analysis is obtained existentially via a fixed-point argument, and is not known to be efficiently computable.


\paragraph{Our approach via Bundle Score Generators.}
Our main technical contribution is a new tool, which we call a \emph{bundle score generator} (BSG), that efficiently attains the symmetric covering property. This tool immediately implies Feige's $(2+\epsilon)$-approximation in the offline setting. It also naturally fits into the online analysis framework of \cite{correa2023constant}, allowing us to recover their $(6+\epsilon)$-competitive ratio efficiently. Finally, by ``opening up'' the analysis of \cite{correa2023constant}, we can attain an improved competitive guarantee for identically distributed buyers. Next we describe the BSG.

A BSG is a randomized mapping that takes an arbitrary bundle $A$ as its input, and outputs a non-negative ``score vector'' supported on the items $A$. 
In the offline setting, after having each buyer $i \in [n]$ draw its tentative bundle $T_i$ independently, 
we apply buyer \(i\)'s BSG to \(T_i\) to obtain a score vector $b_i$, where \(b_{i,j}\) denotes buyer \(i\)'s score for item \(j\in T_i\) (here $b_{i,j} = 0$ for all $j \notin T_i$). Then, for each $j \in [m]$, we assign $j$ to the buyer $i$ with $j \in T_i$ and for which $b_{i,j}$ is largest (breaking ties arbitrarily). 
%

More formally, condition on buyer $i$ drawing a tentative bundle $T_i=A$. Then, the buyer receives the subset $R \subseteq A$ of items $j$ for which its score $b_{i,j}$ exceeds the value $p_j :=\max_{k \ne i}b_{k,j}$. Hence,
\[
    R \supseteq \{j \in A: b_{i,j} > p_j\}.
\]
We construct a specific BSG for which the score vector $b_i$ \emph{strictly stochastically dominates $p$ on $A$}.
Strict stochastic dominance on $A$ means that \(b_i\) can be coupled with a
copy $p'$ of $p$ so that
\(b_{i,j}> p'_j\) for all \(j\in A\) simultaneously (see \Cref{def:stoc-dom}). More, since $b_i$ is independent of
$p$ by the algorithm's definition, we can assume that $p'$ is also independent of $p$.
Now, consider $R_1 := \{j \in A: p'_{j} \ge p_j\}$ and $R_2 := \{j \in A: p_j \ge p'_j\}$.
Clearly, $R_1$ and $R_2$ are identically distributed, as $p$ and $p'$ are independent and identically distributed (i.i.d).
More, $A = R_1 \cup R_2$, and $R_1 \subseteq R$, where the latter holds since 
whenever \(j\in R_1\), we have $p'_j\ge p_j$, and hence
\(b_{i,j}> p'_j\ge p_j\), so $j \in R$. Thus, $R_1, R_2$ give the desired symmetric covering property.


\paragraph{Constructing our BSG.}
As mentioned in the previous subsection, our goal is to generate a score vector $b_i$ for every buyer $i \in [n]$, with the property that, after conditioning on a buyer's tentative set $T_i=A$, the vector $b_i$ strictly stochastically dominates the vector $p$ given as $p_j:=\max_{k\neq i}b_{k j}$ on $A$. This property should hold simultaneously for all the buyers $i \in [n]$.

Our idea is to use \emph{mirroring} to achieve this.
A recurring theme that arises when studying auctions over subadditive valuations is the idea of mirroring: a buyer can win a set of items of high value by simply copying the distribution over bids of the other buyers~\cite{Feldman2015SimultaneousAW,Bei2017WorstCaseMD,neogi2024budget,correa2023constant}.

In order to formalize this, lets assume that each buyer $k =1, \ldots, n-1$ has selected a BSG, say $\psi_k$, where
for each $S \subseteq [m]$, $\psi_{k}(S)$ is a random vector in $\mb{Z}_{\ge 0}^m$ supported on $S$. Suppose now that buyer $n$ must also commit 
to a BSG approach, say $\psi_n$.
Given $(\psi_i)_{i=1}^{n-1}$, consider defining $\psi_n(A)$ on $A \subseteq [m]$ via the following mirroring protocol:
\begin{itemize}
	\item Draw a (simulated) independent tentative subset $T'_k$ for each $k \neq n$, and set $b'_k = \psi_{k}(T'_k)$. 
	\item Set $p'_j = \max_{k \neq n}b'_{k, j}$ for each $j \in [m]$, and define $b_n = \psi_n(A)$ according to \[ b_{n,j} = 
	\begin{cases}
		1+p'_j & \text{if } j \in A. \\
		0 & \text{otherwise.}
	\end{cases} \] 
\end{itemize}
To see why this is a good choice of $\psi_n$ for buyer $n$, notice
that in the actual execution of the algorithm, each buyer $i$ draws a tentative subset $T_i \subseteq [m]$, and then applies $\psi_i$ to $T_i$ to get a scoring vector $b_i = \psi_i(T_i)$. Then, condition on buyer $n$ drawing tentative subset $T_n = A$. If $p_j = \max_{k \neq n}b_{k, j}$ for each $j \in [m]$, then $b_n$ strictly stochastically dominates $p$ on $A$
(due to the $+1$ term in definition of $b_n$). By the previous discussion, this implies the desired value retention for buyer $n$.
To see this explicitly, notice that buyer $n$ receives the items $R \subseteq A$ satisfying:
\begin{equation}
    R \supseteq \{ j \in A: b_{n,j} > \max_{k \neq n} b_{k,j}\} \supseteq \{ j \in A: \max_{k \neq n}b'_{k, j} \ge \max_{k \neq n} b_{k,j}\}.
\end{equation}
If $R_1 =  \{ j \in A: \max_{k \neq n}b'_{k, j} \ge \max_{k \neq n} b_{k,j}\}$ and $R_2 = \{j \in A: \max_{k \neq n}b_{k, j} \ge \max_{k \neq n} b'_{k,j}\}$ then clearly $R_1 \subseteq R$ and $R_1 \cup R_2 = A$. More, since $p$ and $p'$ are i.i.d., it follows that $R_1$ and $R_2$ are identically distributed.
Thus, $R_1$ and $R_2$ form a symmetric covering of $A$, and so we get that $\E[v_{n}(R) \mid T_n = A] \ge v_{n}(A)/2$.

While this only ensures that agent $n$ can attainable a favourable guarantee, we can prove the existence of BSGs $\psi_1,\ldots,\psi_n$ such that each $\psi_i$ is simultaneously obtained by applying the mirroring protocol on $(\psi_k)_{k\neq i}$. By the same analysis, this ensures that every agent attains a favourable guarantee simultaneously, and thus the resulting allocation has high social welfare.
In \Cref{sec:fixed-point-intgap}, we formalise this in terms of a fixed point argument, and in \Cref{sec:MMS} we show an application of similar fixed point arguments for the Maximin share problem over subadditive buyers.
While this fixed point is not necessarily efficiently computable, the above construction for buyer $n$ suggests
a recursive construction of BSGs which \textit{are} efficiently computable. In fact, we are able to define a single BSG $\psi$ that is used by all buyers. The rough idea is as follows: to define $\psi$ on an initial set $A$, we sample a \textit{simulated} tentative bundle $T'_i$ and then recursively apply $\psi$ to the intersection $A \cap T'_i$ for each buyer $i$. Each item $j \in A$ is then assigned
score $1$, plus the maximum score of all $\psi(A \cap T'_i)$ for which $j \in A \cap T'_i$. While this definition appears circular, we leverage the item constraints of the LP solution to argue that the procedure terminates with probability $1$, following a Galton-Watson process style analysis. Ultimately, the score of an item in $A$ ends up being the length of the longest chain of intersections for which it appears in.

\paragraph{Online arrivals.}

The online algorithm uses our BSG-based scoring rule together with an
independent threshold sample. 
Before the buyers arrive, we independently sample a tentative allocation, assign ``base'' scores for this tentative allocation using our BSG, and set each item threshold to the maximum base score assigned to that item.
When buyer \(i\) arrives, we sample its tentative bundle from the LP conditional on its realized valuation, generate scores using our BSG, and allocate only the available items whose scores exceed these thresholds.

The analysis can be summarized as a two-step reduction.
Lemma 1 of Correa and Cristi~\cite{correa2023constant} (which they call the Mirror Lemma) reduces the analysis of the actual online allocation to an offline comparison. Specifically, by forfeiting a factor of \(2\), they show that it suffices to analyze the set of items on which the buyer's score beats two independent
threshold samples. We then apply a ``symmetric $3$-covering''.
The buyer's score vector and the two threshold vectors play symmetric roles,
and subadditivity implies that the buyer retains a third of the tentative value
in this offline comparison. Together, these two losses give the overall $1/6$
factor.

In the setting where buyers' valuations are distributed identically,
we improve the constant $6$ to $60/11$ by refining the analysis, without changing the
algorithm. Rather than applying the Mirror Lemma as a black box, we use the
additional symmetry among buyers to prove a tighter value-retention bound. 
We also identify a natural limitation of this particular combinatorial refinement.

%% file: related_work.tex
\subsection{Further Related Work} \label{sec:related}

There has been extensive work on online allocations problem, motivated from algorithmic and mechanism design applications.


\emph{Prophet inequalities.}
The classical single-item prophet inequality shows that, for independent
nonnegative random variables, a simple threshold rule obtains half the expected
maximum~\cite{krengel1977semiamarts,krengel1978semiamarts,samuelcahn1984threshold}.
Since then, prophet inequalities have been extended to a broad range of
combinatorially constrained selection problems, including matroid and
matroid-intersection constraints~\cite{kleinberg2012matroid,Ehsani2017,jiang2023tight, ArnostiM23},
polymatroid constraints~\cite{dutting2013polymatroid}, matching
constraints~\cite{alaei_2012,ezra2020online, Fu2021, PollnerRSW24, MacRuryMa2024, macrury2025random, Chen25, nuti2026stationary}, and knapsack-type
constraints~\cite{jiang2023tight,yoshinaga2025, ma2025knapsack}. Prophet inequalities have also had
significant applications in mechanism design, where the goal is to design
allocation rules for strategic buyers. Threshold-based algorithms correspond
naturally to posted-price mechanisms: buyers are offered prices and choose their
preferred bundles. In particular, Feldman, Gravin, and
Lucier~\cite{feldman2014combinatorial} gave a constant-factor posted-price
mechanism for XOS valuations, and subsequent work ~\cite{dutting2017prophetinequalitieseasystochastic} gave a polynomial-time
implementation of such mechanisms. For subadditive valuations, D\"utting, Kesselheim and Lucier~\cite{dutting2020log} gave an $O(\log\!\log m)$-competitive static
posted-price mechanism, while obtaining an $O(1)$-competitive item posted-price
mechanism remains an interesting open question.


\emph{Online Contention Resolution Schemes (OCRSs).} 
Contention resolution schemes were originally formalized in the offline setting 
as a tool for constrained submodular optimization \cite{chekuri2014submodular}, motivated in part by the randomized rounding procedures studied in \cite{dobzinski2006, feige2006maximizing, FeigeV10}. Since then, they have been extended
to the online/random-order setting \cite{GuptaN13, feldman2015online, adamczyk2018random, lee2018optimal},
and have become a particularly useful tool for designing prophet inequalities. 
For online combinatorial allocations,  if the buyers have XOS valuations, then OCRSs can be used to attain state-of-the-art competitive ratios, as we explain in \Cref{app:xos-ocrs-reduction}. We note that the usage of offline contention resolution for XOS valuations was recognized by \cite{feige2006maximizing}. However, this paper preceded the formalization of OCRSs, making their usage in online combinatorial allocation problems somewhat folklore, despite similar ideas being used in various other papers \cite{Ehsani2017, chawla2025, MaMZ26Network}. 


\emph{Subadditive optimization.}
Dobzinski, Nisan, and Schapira~\cite{dobzinski2010complementfree} initiated the systematic study of approximation algorithms for combinatorial auctions with subadditive bidders. Feige~\cite{feige2006maximizing} later gave the constant-factor approximation that serves as the offline benchmark for our work. Recently, subadditive valuations have also been studied for Nash social welfare objectives, where \cite{dobzinski2024nsw} gave a constant-factor approximation.

Subadditive valuations have also played a central role in algorithmic mechanism design. For welfare maximization, 
Dobzinski~\cite{Dobzinski-Approx07} gave an $O(\log m\log\!\log m)$-approximate truthful mechanism for subadditive combinatorial auctions, and the current best known truthful mechanism achieves $O((\log\!\log m)^3)$-approximation \cite{assadi2021improved}. 
For revenue maximization, we know $O(1)$-approximation truthful mechanisms \cite{CW-FOCS23,CZ-STOC17,RW-EC15}.
Subadditive valuations are also studied in budget-feasible mechanism design, where a buyer with a hard budget procures items from strategic sellers \cite{Singer10,neogi2024budget,neogi2025budget}. In particular, \cite{neogi2025budget} gave a polynomial-time $O(\log\!\log n)$-approximation truthful mechanism with demand-oracle access to subadditive valuation, and improving it to an $O(1)$-approximation in polynomial-time remains an interesting open question.

%% file: preliminaries.tex
\section{Model and Preliminaries}
\label{sec:preliminaries}

We first recall some basic definitions. For a set of items $[m]$, a valuation is a function
$v:2^{[m]}\to \R_{\ge 0}$. Throughout, we assume valuations are normalized, i.e.,
$v(\emptyset)=0$. We say that $v$ is \emph{monotone} if $v(T) \le v(S)$ for all
$T\subseteq S\subseteq[m]$, and \emph{subadditive} if
\[
    v(S\cup T) \le v(S)+v(T) \qquad \text{for all } S,T\subseteq[m].
\]
A \emph{demand oracle} for valuation $v$ takes as input nonnegative item prices
$c_1,\dots,c_m$ and returns a set
    $S \in \arg\!\max_{T\subseteq[m]}
    \big\{v(T)-\sum_{j\in T} c_j\big\}.$
We assume that demand queries also return the value of the demanded set, or equivalently that value queries are available when needed.

\paragraph{Model.}  In the online combinatorial allocation problem, there is a set of
items $[m]$ and a set of buyers $[n]$. Each buyer $i$ has a valuation
$v_i:2^{[m]}\to \R_{\ge 0}$ drawn independently from a publicly known
distribution $\VD_i$. We assume that each $\VD_i$ is supported on normalized,
monotone, subadditive valuations.
We further assume that each $\VD_i$ has finite support. With probability
$\typeprob_{i,\type}$, the realized valuation is $v_{i,\type}$. We refer to $\type$
as the \emph{type} of buyer $i$, and write $\typeset_i$ for the set of
possible types of buyer $i$; thus,
\[
    \Pr_{v_i \sim \VD_i}[v_i = v_{i,\type}] = \typeprob_{i,\type}
    \qquad \text{for each } \type \in \typeset_i.
\]
For notational simplicity, we index buyers by their arrival order, so buyers arrive
as $1,2,\dots,n$. This also captures any fixed arrival order by relabeling buyers.

Upon the arrival of buyer $i$, the algorithm observes the realized type $\type$,
equivalently the valuation $v_{i,\type}$, and must immediately and irrevocably
allocate to buyer $i$ a subset $ALG_i$ of the remaining items, i.e., $\ALG_i \subseteq [m] \setminus \bigcup_{j=1}^{i-1} \ALG_j$. The algorithm and analysis apply to any fixed order; we relabel buyers according
to that order.
The algorithm's goal is to maximize the expected welfare of the resulting allocation, i.e., $\E\big[\sum_{i\in [n]} v_i(\ALG_i)\big]$. 

Let $(\OPT_i)_{i\in[n]}$ denote the random optimal allocation in hindsight. 
We use the shorthand     $\E[\OPT]
    :=
    \E\big[\sum_{i\in[n]} v_i(\OPT_i)\big]$.
An online algorithm has \emph{competitive factor} $\gamma$ if, for every collection
of input distributions $(\VD_i)_{i\in[n]}$, we have
  $  \E\big[\sum_{i\in[n]} v_i(\ALG_i)\big]
    \ge
    \frac{1}{\gamma}\,\E[\OPT]$.
Equivalently, it is $\gamma$-competitive if the left-hand
side is at least $(1/\gamma)\E[\OPT]$. Analogously, an offline algorithm is $\gamma$-approximate if the allocation $(\ALG_i)_{i\in[n]}$ it returns has value at least $1/\gamma$ fraction of the optimal offline value.



\paragraph{Computational access.}
The distributions $\VD_i$ are known explicitly through their finite supports and
probabilities. In preprocessing, the algorithm has demand oracle access to every valuation
$v_{i,\type}$ in the support. When buyer $i$ arrives, its realized type is revealed.
All running-time statements are measured in $n$, $m$, the total number of types,
and the relevant accuracy parameters.

\paragraph{Ex-ante configuration LP.}
For a parameter $\alpha\in[0,1]$, consider the standard ex-ante
configuration LP, denoted $LP(\alpha)$:
\begin{align} \label{lp:subadditive}
\tag{$\mathrm{LP(\alpha)}$}
\max \quad
& \sum_{i\in[n]} \sum_{\tau\in\Gamma_i} \sum_{\emptyset\ne S \subseteq[m]} x_{i,\tau,S} \cdot v_{i,\tau}(S) \notag \\
\text{s.t.} \quad
& \sum_{\emptyset\ne S\subseteq[m]} x_{i,\tau,S} \le \typeprob_{i,\tau}
    && \forall i\in[n],\ \tau\in\Gamma_i, \label{eq:subadd-lp-type}\\
& \sum_{i \in [n]}\sum_{\tau\in\Gamma_i} \sum_{S \ni j} x_{i,\tau,S} \le \alpha
    && \forall j\in[m], \label{eq:subadd-lp-marginal}\\
& x_{i,\tau,S} \ge 0
    && \forall i\in[n],\ \tau\in\Gamma_i,\ S\subseteq[m]. \notag
\end{align}

Let $\mathrm{LP(\alpha)}^*$ denote the optimal value to this linear program.
We use the standard fact that the ex-ante configuration LP, despite having
exponentially many variables, can be optimized using demand-oracle access via
dual separation and the ellipsoid method to return a polynomial-support primal solution
\cite{nisan2007algorithmic,dutting2017prophetinequalitieseasystochastic}.
 As usual, one may compute a
$(1-\eta)$-approximate solution in time polynomial in the input size and
$\log(1/\eta)$; this only changes our final guarantees by a factor of $1-\eta$.
In the sequel, we suppress this approximation issue and write as though exact
optimal LP solutions are available; whenever a final guarantee is stated with
slack $\epsilon$, we choose $\eta$ sufficiently small and absorb the resulting
$(1-\eta)$ loss into $\epsilon$, with only a further $\poly(1/\epsilon)$ loss in
the running time.

\begin{restatable}{fact}{LPProp}
\label{clm: LP-OPT}
For every $\alpha\in[0,1]$, we have 
    $\mathrm{LP(\alpha)}^* \ge \alpha\,\E[\OPT]$.
Moreover, for any $\eta>0$, a polynomial-support $(1-\eta)$-approximate
solution to \ref{lp:subadditive} can be computed in time polynomial in
$n$, $m$, $\sum_i|\typeset_i|$, and $\log(1/\eta)$, given demand-oracle
access to every valuation $v_{i,\type}$.
\end{restatable}

%% file: recursive_procedure.tex
\section{Bundle Score Generators}
\label{sec:bsg}


In this section, we introduce our bundle score generator
(BSG) definition. We then give our main construction, the Recursive-Height BSG, and prove a stochastic dominance property needed to prove our results.

\subsection{Motivation and definition}

To motivate BSGs, consider first the offline allocation problem, where
each buyer $i$ has a fixed valuation $v_i$ and the goal is to allocate
the items $[m]$ among the buyers. This is captured by \ref{lp:subadditive} of \Cref{sec:preliminaries} by setting
$\typeset_i=\{v_i\}$ and $\typeprob_{i,v_i}=1$, so that the LP variables reduce
to $x_{i,A}$ for $i\in[n]$ and $\emptyset\neq A\subseteq[m]$.


A natural rounding strategy used in~\cite{feige2006maximizing}
is to first sample a \emph{tentative bundle} $T_i$ for each buyer $i$, with
$\Pr[T_i=A]=x_{i,A}$ for each $A \subseteq [m]$ (this is well-defined, due to \eqref{eq:subadd-lp-type} of \ref{lp:subadditive}). If the tentative bundles were disjoint, the expected
welfare of this tentative allocation would match the LP objective exactly. Thus, the central issue is
to resolve the conflicts over those items which appear in multiple tentative bundles,
while ensuring the allocation attains a constant fraction of the LP value's solution.


A bundle score generator provides such a rule. Given a tentative bundle, it
assigns possibly correlated random scores to the items in the bundle. We formalize our definition and notation below: 

\begin{definition}[Bundle score generator]
A \emph{bundle score generator} (BSG) is a randomized map $\psi$ that inputs a subset $A\subseteq[m]$, outputs a \textit{random score vector}
$\psi(A)\in \Z_{\ge 0}^{A}$, where $\psi_j(A)$ is the \textit{score} assigned to
item $j\in A$. The coordinates of $\psi(A)$ may be arbitrarily correlated.
\end{definition}
Given a BSG, our ``conflict resolution'' template for the offline allocation problem is simple. After having each buyer $i \in [n]$ draw a tentative bundle $T_i$, apply the BSG to obtain scoring vectors $(\psi(T_i))_{i \in [n]}$. Then, for each $j \in [m]$, we assign $j$ to the agent $i$
with $j \in T_i$ and for which $\psi_{j}(T_i)$ is largest (breaking ties arbitrarily). Our goal is to design a BSG for which this conflict resolution template has expected welfare at least $\mathrm{LP(\alpha)}^*/2$. We achieve this by designing a specific ``Recursive-Height BSG'' which satisfies the ``symmetric covering'' property of \Cref{sec:techniques}.




\subsection{Stochastic Dominance and Subcriticality}
In this section, we design a Recursive-Height BSG which satisfies a certain
stochastic dominance property. As we shall see, this stochastic dominance will allow us to attain the desired self-complementarity property. 
We use the following notation. For subsets $B\subseteq A\subseteq[m]$ and
a vector $y\in\Z_{\ge0}^{A}$, let $y|_B\in\Z_{\ge0}^{B}$ denote the
restriction of $y$ to $B$, i.e., $(y|_B)_j=y_j$ for every $j\in B$.

\begin{definition}[Strict stochastic dominance]
\label{def:stoc-dom}
    For random vectors $X,Y\in\Z_{\ge0}^{A}$, we write $X\succ Y$ if
    there exists a coupling of $X$ and $Y$ such that  $X_j>Y_j$ for all $j \in A$ simultaneously.
\end{definition}

The Recursive-Height BSG is defined from a collection of
distributions over subsets of $[m]$. In our applications, these tentative subsets are sampled from a solution to \ref{lp:subadditive} for $0 < \alpha < 1$.

\begin{definition}[Subcritical  distributions]
    A collection of distributions $\{\nu_i\}_{i\in[n]}$ over subsets of
    $[m]$ is $\alpha$-subcritical, for some $0<\alpha<1$, if for every
    item $j\in[m]$,
    \begin{align}
    \label{eq: subcritical}
        \sum_{i\in[n]}
        \Pr_{L_i\sim\nu_i}[j\in L_i]
        \le \alpha .
    \end{align}
\end{definition}

The distributions arising from the configuration LP rounding are
subcritical in exactly this sense. Given any feasible solution
$\overline{x}$ to \ref{lp:subadditive}, define distributions
$\{\mu_i\}_{i\in[n]}$ by
    \begin{align}
    \label{eq: mu_def}
        \Pr_{L_i\sim\mu_i}[L_i=A]
        =
        \begin{cases}
            \displaystyle
            \sum_{\type\in\typeset_i}\overline{x}_{i,\type,A},
                & \text{if } A\neq\emptyset, \\[1ex]
            \displaystyle
            1-\sum_{\type\in\typeset_i}\sum_{R\neq\emptyset}
                \overline{x}_{i,\type,R},
                & \text{if } A=\emptyset.
        \end{cases}
    \end{align}
    Then, by \eqref{eq:subadd-lp-marginal}, the collection
    $\{\mu_i\}_{i\in[n]}$ is $\alpha$-subcritical.

\subsection{The Recursive-Height BSG} \label{sec:recHeightBSG}

Fix any $\alpha$-subcritical collection $\{\nu_i\}_{i\in[n]}$.
We define the Recursive-Height BSG $\Height(\cdot)$ as follows: 
For the empty set, $\Height(\emptyset)$ returns the empty vector; for a
nonempty set $A\subseteq[m]$, $\Height(A)$ is sampled recursively as in
\Cref{alg:recHeight}.

\begin{algorithm}[tbh]
\caption{$\Height(A)$}\label{alg:recHeight}
\begin{enumerate}[nolistsep]
    \item For each $i\in[n]$, independently sample a random interest set
    $L_i\sim\nu_i$.

    \item Let $J_i:=L_i\cap A$. If $J_i\neq\emptyset$, recursively sample
    an independent vector
    \[
        y_i:=\Height(J_i)\in\Z_{\ge0}^{J_i}.
    \]

    \item Output the vector $b\in\Z_{\ge0}^{A}$ defined by
    \[
        b_j
        =
        1+
        \max\Bigl(
            \{0\}
            \cup
            \{(y_i)_j : i\in[n],\, j\in J_i\}
        \Bigr)
        \qquad \forall j\in A.
    \]
\end{enumerate}
\end{algorithm}

\textit{Remark.} The recursive algorithm above mimics the mirroring protocol described in \Cref{sec:techniques}, where now we use recursive subcalls to compute the simulated scoring vector of each buyer. We note that unlike the mirroring principle -- where we left out an arbitrary buyer $n$ -- here we sample \textit{every} buyer (including $n$). This is to simplify the construction and analysis. We get a single BSG, opposed to different one for each buyer.

\begin{theorem}[Dominance of the \Height{} BSG]
\label{thm:bsg-guarantees}
    Given any $\alpha$-subcritical collection of distributions
$\{\nu_i\}_{i\in[n]}$ for some $0<\alpha<1$, the Recursive-Height BSG
$\Height(\cdot)$ satisfies the following:
    \begin{itemize}
        \item  $\Height(\cdot)$ can be implemented in expected time polynomial in $n$, $m$, and $(1-\alpha)^{-1}$.
    \item 
    For any $A\subseteq[m]$, let $b=\Height(A)$. Independently draw
    samples $L_k'\sim\nu_k$ and $b_k'=\Height(L_k')$ for each $k\in[n]$.
    Define,  for every $j\in[m]$, 
        \[p'_j
        :=
        \max\Bigl(
            \{0\}
            \cup
            \{b'_{k,j}: k\in[n],\, j\in L_k'\}
        \Bigr).\]
    Then,
       $ b\succ p'|_A .$
    \end{itemize}
\end{theorem}

\subsection{Proof of \Cref{thm:bsg-guarantees}}

We first show that the $\Height$ procedure terminates almost surely -- from the definition of the $\Height$ procedure, there are sample paths such that the procedure may never terminate. However, the measure of all such sample paths is $0$:







\begin{lemma}
\label{clm:polytime}
    The procedure $\Height(A)$ terminates in a finite number of recursive calls almost surely, for any input tentative bundle $A \subseteq [m]$. Further, the expected total size of interest sets in the recursion tree is at most $|A|/(1-\alpha)$. 
\end{lemma}

\begin{proof}
    Let $Z_i$ denote the sum of sizes of the interest sets in the recursion tree at depth $i$ with $Z_0 = |A|$. Then, for a recursive node with interest set $A'$, the expected sum of sizes of the interest sets for its children is:
    \[
        \E\left[\sum_{i=1}^{n}|A' \cap L_i|\right] = \sum_{j \in A'}\sum_{i=1}^{n}\Pr_{L_i \sim \nu_i}[L_i \ni j] \le \alpha \cdot |A'|
    \]
    In the above, we use the subcriticality of the distributions $\nu_i$ (Eq \ref{eq: subcritical}). So, $\E[Z_i|Z_{i-1}] \le \alpha \cdot Z_{i-1}$, and thus $\E[Z_i] \le \alpha^i \cdot |A|$. Summing over all $r$, $\E[\sum_{i}Z_i] \le |A|/(1-\alpha)$. Since the expected total size of interest sets in the recursion tree is finite, the number of nodes in the recursion tree is finite almost surely. 
\end{proof}

\emph{Remark.}
By Markov's inequality and a union bound, the preceding bound allows
any polynomial number of calls to $\Height(\cdot)$ to be capped so
that the probability of any truncation is at most $\delta$, while
keeping the running time polynomial in the input size and $1/\delta$.
We suppress explicit truncation in what follows: a standard coupling
argument shows that, after suitably rescaling the truncation
probability, the value retention guarantees established below incur
at most a factor $(1-\delta)$ loss. Absorbing this loss into the
slack $\epsilon$, every expected-polynomial-time guarantee we state
can therefore be converted into a worst-case polynomial-time
guarantee, with an additional $\poly(1/\epsilon)$ dependence in
the running time.




We will also need the following fact about the distribution of restricted vectors:

\begin{lemma}
\label{clm:rest-consistency}
    Consider any $S \subseteq A \subseteq [m]$, and let $b_A = \Height(A), b_S = \Height(S)$. Then, $b_A|_S \overset{d}{=} b_S$, where $\overset{d}{=}$ denotes equality in distribution. 
\end{lemma}

With these lemmas, we first finish the proof of \Cref{thm:bsg-guarantees}:

\begin{proof}[Proof of \Cref{thm:bsg-guarantees}]
    The running time guarantee of $\Height(\cdot)$ is immediate from \Cref{clm:polytime} and \Cref{clm: LP-OPT} (due to which sampling the interest subsets can be done in polynomial time). Now, consider some fixed subset $A$. We identically couple the random sets $L_k'$ with the random sets $L_k$ sampled in the first level of the recursion in $\Height(A)$ for $b$, i.e., under this coupling, $L_k = L_k'$. Recall that $J_k = L_k \cap A = L_k' \cap A$ and $y_k = \Height(J_k)$. Then, by definition:
    \[
        b_j = 1 + \max\Bigl(\{0\} \cup \{(y_k)_j: k \in [n], j \in J_k\}\Bigr)
    \]
    Since $b_k' = \Height(L_k') = \Height(L_k)$, by \Cref{clm:rest-consistency}, $b_k'|_{J_k} \overset{d}{=} y_k$. By definition, for $j \in A$
    \[
        (p'|_A)_j = \max\Bigl(\{0\} \cup \{(b_k'|_{J_k})_j: k \in [n], j \in J_k\}\Bigr)
    \]
    Therefore:
    \[
        b \overset{d}{=} \mathbf{1} + p'|_A
    \]
    where $\mathbf{1} \in \Z^{A}$ is the all $1$'s vector. Hence, it is clear to see that $b \succ p'|_A$. 
\end{proof}

It remains to prove \Cref{clm:rest-consistency}:

\begin{proof}[Proof of \Cref{clm:rest-consistency}]
    We couple the procedures $\Height(A)$ and $\Height(S)$ by using the same
    sampled sets $L_1,\dots,L_n$ at every corresponding recursive call. Consider the output $b_A$ of a run of $\Height(A)$ which has finite termination depth, i.e., we condition on all the random interest subsets sampled at every node of the recursion. 
    
    Then, $b_A|_S$ is obtained by simply dropping all coordinates of $b_A$ outside $S$. Equivalently, $b_A|_S$ is the output vector obtained by restricting every interest subset of $\Height(A)$ to $S$, since the values for any coordinates outside $S$ are never used to define $(b_A)|_S$. For any child node of the root of the recursion with interest set $L_k$, the restricted interest set now becomes $(L_k \cap A) \cap S = L_k \cap S$, which is identical to the restricted interest set in the coupled run of $\Height(S)$. 
    
    Applying the argument inductively on the depth of the recursion, we get the equality of all restricted interest sets for this run and thus, $b_A|_S = b_S$ under the coupling. By \Cref{clm:polytime}, since the recursion terminates in finite time almost surely for any input subset $S$, we have $b_A|_S \overset{d}{=} b_S$. 
\end{proof}

%% file: offline_rounding.tex
\section{$(2+\epsilon)$-Approximation Offline Algorithm}\label{sec:offline_alg}

We now analyze the offline rounding approach discussed in \Cref{sec:bsg} and show that it achieves an approximation ratio of $1/(2+\epsilon)$ for any $\epsilon > 0$. For clarity, we formally present the algorithm in \Cref{alg:offline-rounding}. It is clear to see that the algorithm produces a feasible allocation. 

\begin{algorithm}[tbh]
\LinesNotNumbered
\caption{Offline rounding algorithm}
\label{alg:offline-rounding}
\begin{minipage}{0.96\textwidth}
Given a feasible solution $\overline{x}$ to \ref{lp:subadditive}, let $\Height(\cdot)$ be
the corresponding \SHSG{} procedure, defined using the induced distributions
$\{\mu_i\}_{i\in[n]}$ from  \eqref{eq: mu_def}.

\begin{enumerate}
    \item Each buyer $i\in[n]$ independently samples a tentative bundle $T_i \sim \mu_i$, and then generates a \bid{} vector $b_i := \Height(T_i)$.


    \item For each item $j\in[m]$ that appears in at least one tentative bundle,
    assign $j$ to a buyer $i \in \arg\max\{b_{i,j} : j\in T_i\}$,
    breaking ties arbitrarily. Items appearing in no tentative bundle may be left
    unallocated.
\end{enumerate}
\end{minipage}
\end{algorithm}

The following theorem captures the buyer-wise value retention guarantee, which immediately implies approximation ratio as in \Cref{thm: offline-factor-2}.
\begin{lemma}
\label{thm: offline-val-pres}
    For buyer $i \in [n]$ with valuation $v_i$, for any subset $\emptyset \ne A \subseteq [m]$ with $x_{i, A} > 0$, if $ALG_i \subseteq A$ is the final subset of items assigned to buyer $i$ by \Cref{alg:offline-rounding}, then
    \[
        \E[v_i(\ALG_i) \mid T_i = A] \ge \frac{1}{2}v_i(A)
    \]
\end{lemma}

\begin{proof}
    Condition on $T_i=A$. For each item $j\in[m]$, let
    \[
        p_j := \max\Bigl(\{0\}\cup\{b_{k,j}: k\neq i,\ j\in T_k\}\Bigr)
    \]
    be the maximum score assigned to item $j$ by buyers other than $i$. Buyer $i$ receives every item $j\in A$ for which $b_{i,j}>p_j$, and possibly additional items due to tie-breaking. 
    We first complete $p$ into a full threshold score $\tilde p$ by adding an independent hallucinated copy of buyer $i$. Draw
    $\widetilde{T}_i\sim\mu_i$ independently, set
    $\widetilde{b}_i:=\Height(\widetilde{T}_i)$, and define
    \[
        \widetilde{p}_j
        :=
        \max\Bigl(
            \{p_j\}
            \cup
            \{\widetilde{b}_{i,j}: j\in \widetilde{T}_i\}
        \Bigr)
        \qquad \forall j\in[m].
    \]
    Since $\widetilde{p}_j\ge p_j$ for every $j$, monotonicity of $v_i$ gives
    \begin{equation}
        \E[v_i(\ALG_i)\mid T_i=A]
        \ge
        \E\!\left[
            v_i\bigl(\{j\in A: b_{i,j}>p_j\}\bigr)
        \right]
        \ge
        \E\!\left[
            v_i\bigl(\{j\in A: b_{i,j}>\widetilde{p}_j\}\bigr)
        \right]. \label{eqn:additional_copy_of_i}
   \end{equation}
    As noted before, the distributions $\{\mu_i\}_{i\in[n]}$ are
    $\alpha$-sub-critical. Moreover, conditional on $T_i=A$, the vector
    $b_i=\Height(A)$ is independent of the threshold vector
    $\widetilde{p}$. Therefore, by \Cref{thm:bsg-guarantees}, $b_i \succ \widetilde{p}|_A$.  
    By the stochastic dominance guarantee, we can couple $b_i$ with an additional random vector
    $p' \in \R^A$ such that 
    \begin{itemize}
        \item  $b_{i,j}>p_j'$ for all $j\in A$ simultaneously,
        \item $p'$ has the same law as $\widetilde{p}|_A$, and is drawn independently from it.
    \end{itemize}
    Observe then that under this coupling,
    $
    \{j\in A: b_{i,j}>\widetilde{p}_j\}  \supseteq \{j\in A: p_j'\ge \widetilde{p}_j\}
    $
    and so using monotonicity,
    $
       v_i\bigl(\{j\in A: b_{i,j}>\widetilde{p}_j\}\bigr) \ge v_i(\{j\in A: p_j'\ge \widetilde{p}_j\}).
    $
Thus, after taking expectations and applying \eqref{eqn:additional_copy_of_i},
    \begin{align}
        \E[v_i(\ALG_i)\mid T_i=A]
        &\ge
        \E
        \left[
            v_i\bigl(\{j\in A: b_{i,j}>\widetilde{p}_j\}\bigr)
        \right] \notag \\
        &\ge
        \E\left[v_i(\{j\in A: p_j'\ge \widetilde{p}_j\})\right] \label{eqn:coupling_lower_bound}
    \end{align}
To complete the proof, define the random sets
    \[
        A_0 := \{j\in A: p_j'\ge \widetilde{p}_j\},
        \qquad
        A_1 := \{j\in A: \widetilde{p}_j\ge p_j'\}.
    \]
Then, $A_0$ and $A_1$ satisfy the symmetric covering property described in \Cref{sec:techniques}.
Concretely, observe that $A=A_0\cup A_1$, and so by subadditivity,
    \begin{equation} \label{eqn:subadditivity_offline}
        v_i(A_0)+v_i(A_1)\ge v_i(A).
    \end{equation}
On the other hand, $p'$ and $\widetilde{p}|_A$ are i.i.d., and so the random sets $A_0$ and $A_1$ are identically distributed, which implies $\E[v_i(A_0)] =  \E[v_i(A_1)]$. 
Therefore, by combining this with \eqref{eqn:subadditivity_offline} and taking expectations,
\begin{equation} \label{eqn:cover_offline_lower_bound}
    \E[v_i(A_0)] \ge v_i(A)/2.
\end{equation}
Finally, using \eqref{eqn:coupling_lower_bound} with \eqref{eqn:cover_offline_lower_bound} implies
    $
        \E[v_i(\ALG_i)\mid T_i=A]
        \ge
        \E[v_i(A_0)]
        \ge
        \frac{1}{2}v_i(A)
    $
    as required.
\end{proof}

\begin{theorem}
\label{thm: offline-factor-2}
    For any feasible solution $\overline{x}$ to \ref{lp:subadditive}, \Cref{alg:offline-rounding} produces a feasible allocation with expected welfare at least
    \[
        \frac{1}{2}\sum_{i,S}\overline{x}_{i,S}v_i(S).
    \]
    Its expected running time is polynomial in $n$, $m$, and $(1-\alpha)^{-1}$.
\end{theorem}
\begin{proof}
    By \Cref{thm: offline-val-pres}, we have that $\E[\ALG]$ equals
    \begin{align*}
         \sum_i\E[v_i(\ALG_i)] & ~=~ \sum_{i, S} \Pr[T_i = S] \cdot \E[v_i(\ALG_i)|T_i = S]  ~\ge~ \frac{1}{2}\sum_{i, S} \overline{x}_{i, S} \cdot v_i(S), 
    \end{align*}
    where we used $\Pr[T_i = S] = \overline{x}_{i, S}$. The running-time bound follows from \Cref{thm:bsg-guarantees}, since each call to $\Height(\cdot)$ has expected running time polynomial in $m$ and $(1-\alpha)^{-1}$, and the algorithm makes one such call per buyer.
\end{proof}


By \Cref{clm: LP-OPT}, given demand-oracle access, we can compute in polynomial time a feasible solution $\overline{x}$ to \ref{lp:subadditive} whose objective value is at least $\alpha\cdot \mathrm{OPT}$.
For any $\epsilon > 0$, setting $\alpha=\frac{2}{2+\epsilon}$, 
and applying \Cref{thm: offline-factor-2}, 
we obtain a $(2+\epsilon)$-approximate algorithm with expected running time polynomial in $n$, $m$, and $1/\epsilon$.

%% file: GW_RSG.tex
\section{$(6+\epsilon)$-Competitive Online Algorithm}

\label{sec: online}

In this section, we use the $\Height$ \SHSG{} from \Cref{sec:bsg} to obtain an expected-polynomial-time online algorithm with an $O(1)$ competitive ratio for subadditive valuations. We first begin in \Cref{sec: value-retention} with the key underlying value retention lemma that enables us to obtain the constant competitive ratio. Then, we present the algorithm for the case when the valuation distributions have finite supports in \Cref{sec: finite-type}, and finally show how the algorithm can be extended to the sample access setting in \Cref{sec: infinite-type}.

\subsection{Value Retention Guarantee}
\label{sec: value-retention}

We first state the key structural lemma that enables us to argue the desired value retention. The lemma only exploits the stochastic dominance property of the $\Height(\cdot)$ procedure (for any given collection of sub-critical distributions) and monotonicity and subadditivity of the valuation. 

\begin{lemma}
\label{lem:one-sixth-val}
Consider $\Height(\cdot)$ defined with respect to a sub-critical
collection $(\nu_k)_{k\in[n]}$. Fix $A\subseteq[m]$ and a monotone
subadditive valuation $v$, and draw $b_i\sim\Height(A)$.
Independently, for each $k\in[n]$, draw $T_k'\sim\nu_k$ and
$b_k'\sim\Height(T_k')$
and set $p_j'=\max\bigl(\{0\}\cup
    \{b_{k,j}':k\in[n],\ j\in T_k'\}\bigr)$.
Let $p''$ have the same distribution as $p'$, independently of
$(b_i,p')$. Then,
\[
    \E[v(\{j\in A:b_{i,j}>p_j'\ge p_j''\})]
    \ge \frac16 v(A).
\]
\end{lemma}

\begin{proof}[Proof of \Cref{lem:one-sixth-val}]
    Let $A' \subseteq A$ be
    \[
        A' = \{j \in A: b_{i, j} > p_j' \ge p_j''\}
    \]
    Since the distributions $\{\nu_k\}_{k\in[n]}$ are $\alpha$-sub-critical,
    \Cref{thm:bsg-guarantees} gives $b_i\succ p'|_A$. By the stochastic dominance guarantee, we can couple $b_i$ with an additional random vector
    $\widehat p \in \mathbb{Z}^{A}_{\ge 0}$ such that 
    \begin{itemize}
        \item  $b_{i,j}>\widehat p_j$ for all $j\in A$ simultaneously,
        \item $\widehat p$ has the same law as $p'|_A$ and $p''|_A$, and is drawn independently from them.
    \end{itemize}
    Observe then that under this coupling,
    $
    A' = \{j\in A: b_{i,j}>p'_j\ge p_j''\}  \supseteq \{j\in A: \widehat p_j\ge p'_j\ge p_j''\}. 
    $
    So using monotonicity, under the coupling,
    $
       v(A') \ge v(\{j\in A: \widehat p_j\ge p'_j\ge p_j''\}).
    $ Thus, after taking expectations, 
    \begin{align}
    \label{eq: val-lb}
        &\E\!\left[
            v(A')
        \right]\ge
        \E[
        v(\{j\in A: \widehat p_j\ge p'_j\ge p_j''\})
        ]
    \end{align}

    It remains to analyze the final expectation. 
    For every ordered triple $(u,v,w)$
    obtained by permuting the three vectors $\widehat p$, $p'|_A$, and
    $p''|_A$, define
    \[
        A_{u,v,w}
        :=
        \{j\in A: u_j\ge v_j\ge w_j\}.
    \]
    These six sets cover $A$, since for every item $j\in A$, the three numbers
    $\widehat p_j$, $p_j'$, and $p_j''$ satisfy at least one of the six weak
    orderings. Hence, by subadditivity,
    \begin{align}
    \label{eq: subadd-six}
        \sum_{(u,v,w)} v(A_{u,v,w})
        \ge
        v(A),
    \end{align}
    where the sum ranges over the six ordered triples.

    On the other hand, $\widehat p$, $p'|_A$, and $p''|_A$ are i.i.d., the six random
    sets $A_{u,v,w}$ are identically distributed, and thus each have the same expected value with respect to $v$. On the other hand, notice that
    $A_{\widehat p,p'|_A,p''|_A}=\{j\in A: \widehat p_j\ge p'_j\ge p_j''\}$, and therefore using \eqref{eq: subadd-six},
    \[
        \E[v(A_{\widehat p,p'|_A,p''|_A})]
        =
        \frac{1}{6}
        \E\! \left[
            \sum_{(u,v,w)} v(A_{u,v,w})
        \right]
        \ge
        \frac{1}{6}v(A) ,
    \]
    which when combined with \eqref{eq: val-lb} completes the proof. 
\end{proof}

\subsection{Valuation Distributions with Finite Supports}
\label{sec: finite-type}

We now describe our approach for the setting where each valuation distribution is given explicitly to obtain an ex-ante guarantee against the benchmark \ref{lp:subadditive}.

Given our Recursive-Height BSG from \Cref{sec:recHeightBSG}, the online algorithm follows the random-score perspective of Correa and Cristi~\cite{correa2023constant}: buyers are assigned item scores, and items are allocated only when a buyer's score clears an appropriate base threshold. The distinction is in how these scores are produced. Rather than assuming access to buyer-specific random score generators obtained nonconstructively, our algorithm uses the LP solution to sample tentative bundles and then applies the \SHSG{} to generate correlated scores from each tentative bundle. The base thresholds are generated by running this same scoring procedure once in advance on an independent threshold sample: tentative bundles are drawn from the LP-induced distributions, scored using $\Height(\cdot)$, and each item threshold is set to the maximum score assigned to that item. Thus the LP proposes bundles, the \SHSG{} assigns correlated scores within those bundles, and the independent threshold sample provides the benchmark that arriving buyers must beat.

We formally present the algorithm below. 

\begin{algorithm}[tbh]
\LinesNotNumbered
\caption{Online algorithm $\mathcal{A}$}
\label{alg:online-algorithm}
\begin{minipage}{0.96\textwidth}
Given a feasible solution $\overline{x}$ to \ref{lp:subadditive},
let $\Height(\cdot)$ be the corresponding \SHSG{} procedure, defined using the
induced distributions $\{\mu_i\}_{i\in[n]}$ from \eqref{eq: mu_def}.

\begin{enumerate}
    \item \textbf{Compute base thresholds.}
    Before any buyer arrives, independently sample $T_i' \sim \mu_i$ for each
    $i\in[n]$, and set $b_i' := \Height(T_i')$. 
    Define the threshold vector $p'$ coordinatewise:
    \[
        p_j'
        :=
        \max\Bigl(\{0\}\cup\{b_{i,j}' : i\in[n],\ j\in T_i'\}\Bigr)
        \qquad \forall j\in[m].
    \]

    \item \textbf{Online phase.}
    When buyer $i$ arrives and reveals valuation $v_i=v_{i,\type}$:
    \begin{enumerate}
        \item Sample a tentative bundle $T_i$ according to $\Pr[T_i=A] = \overline{x}_{i, \type, A}/\typeprob_{i, \type}$. 
        Set $b_i := \Height(T_i)$.

        \item Let $S_i := \{j\in T_i : b_{i,j} > p_j'\}$.
        If $R_i$ denotes the set of items remaining when buyer $i$ arrives,
        allocate $\ALG_i := S_i \cap R_i$
        to buyer $i$.
    \end{enumerate}
\end{enumerate}
\end{minipage}
\end{algorithm}

The feasibility of \Cref{alg:online-algorithm} follows immediately since it only allocates remaining items.
We now formally state its performance guarantee with respect to the LP objective. 

\begin{theorem}
\label{thm:sampling-main}
    Given a feasible solution $\overline{x}$ to \ref{lp:subadditive} for some $0 < \alpha < 1$, \Cref{alg:online-algorithm} produces a feasible allocation whose expected welfare is at least:
    \[  
        \frac{1}{6}\sum_{i \in [n]}\sum_{\type \in \typeset_i}\sum_{A \ne \emptyset}\overline{x}_{i, \type, A}v_{i, \type}(A)
    \]
    Moreover, the algorithm has expected running time polynomial in
    $n$, $m$, and $(1-\alpha)^{-1}$.
\end{theorem}

We show the following buyer-by-buyer conditional guarantee about the online algorithm -- for any buyer, conditioned on its type and tentative subset, the expected value of the final set it receives is at least a $1/6$ fraction of the value of the tentative subset. 

\begin{lemma}
\label{lem:conditional-one-sixth}
    For \Cref{alg:online-algorithm}, for any buyer $i$, conditioned on their type $v_{i, \type}$ and random tentative subset $T_i = A$, we have
    \[
        \E[v_i(\ALG_i) \; | \; v_i = v_{i, \type}, T_i = A] \ge \frac{1}{6}v_{i, \type}(A).
    \]
    In the above, the expectation is over the randomness in the types of the previous buyers, their random scores, and the threshold scores. 
\end{lemma}

With this lemma, we first complete the proof of \Cref{thm:sampling-main}. 

\begin{proof}[Proof of \Cref{thm:sampling-main}]
    The expected running time follows directly from \Cref{thm:bsg-guarantees}, since the algorithm makes polynomially many calls to $\Height(\cdot)$. Conditioning on the random valuation for buyer $i$ being $v_{i, \type}$ and the random choice of the tentative subset $T_i = A$, by \Cref{lem:conditional-one-sixth}, we have
    \[
        \E[v_i(\ALG_i)|v_i = v_{i, \type}, T_i=A] \ge \frac{1}{6}v_{i, \type}(A).
    \]
    Taking expectation over the random choice of $T_i$ and type of buyer $i$, we have
    \[  
        \E[v_i(\ALG_i)] \ge \frac{1}{6} \sum_{\type \in \typeset_i}\typeprob_{i, \type}\sum_{A \subseteq [m]}\frac{\overline{x}_{i, \type, A}}{\typeprob_{i, \type}}v_{i, \type}(A) = \frac{1}{6}\sum_{\type \in \typeset_i}\sum_{A \ne \emptyset}\overline{x}_{i, \type, A}v_{i, \type}(A).
    \]
    Summing over all buyers completes the proof. 
\end{proof}

By \Cref{clm: LP-OPT}, given demand-oracle access, we can compute in polynomial time a feasible solution $\overline{x}$ to \ref{lp:subadditive} whose objective value is at least $\alpha\cdot\E[\mathrm{OPT}]$. 
Setting $\alpha = 6/(6+\epsilon)$, and applying \Cref{thm:main}, we obtain a $(6+\epsilon)$-competitive online algorithm with expected running time polynomial in $n$, $m$, $1/\epsilon$, and the total support size of the valuation distributions.


It remains to prove \Cref{lem:conditional-one-sixth}. Our algorithm follows the random-score framework of Correa and Cristi~\cite{correa2023constant}, and consequently, their Mirror Lemma can be applied to obtain the factor-$2$ online loss. The remaining value-retention argument is specific to our \SHSG{} construction: strict stochastic dominance reduces it to a competition for maximality among three i.i.d. threshold vectors, which yields the additional factor $1/3$. 

For completeness, and to make the constants explicit, we give a self-contained proof of \Cref{lem:conditional-one-sixth} below that does not use the Mirror Lemma. The proof also isolates the step that will be sharpened in the i.i.d. setting in \Cref{sec: iid}. 

\begin{proof}[Proof of \Cref{lem:conditional-one-sixth}]
    Condition throughout on $v_i=v_{i,\type}$ and $T_i=A$. Under this
    conditioning, buyer $i$ generates the score vector $b_i=\Height(A)$, and its
    candidate set is
    \[
        S_i=\{j\in A: b_{i,j}>p_j'\}.
    \]

    We first describe the set $R_i$ of items still available when buyer $i$
    arrives. Let $b^{<i}$ be the maximum score assigned to each item by buyers
    arriving before $i$:
    \[
        b^{<i}_j
        :=
        \max\Bigl(\{0\}\cup\{b_{k,j}: k<i,\ j\in T_k\}\Bigr)
        \qquad \forall j\in[m].
    \]
    An item $j$ has already been allocated before buyer $i$ arrives if and only
    if some earlier buyer assigned it score strictly larger than the threshold
    $p_j'$. Therefore
    \[
        R_i=\{j\in[m]: p_j'\ge b^{<i}_j\}.
    \]
    
    Since $\ALG_i=S_i\cap R_i$, buyer $i$ receives exactly
    \[
        \ALG_i
        =
        \{j\in A: b_{i,j}>p_j'\ge b^{<i}_j\}.
    \]

    We next introduce a second threshold vector $p''$. This vector can be viewed as the threshold buyer $i$ would face in the pessimistic scenario where, after the actual previous buyers have arrived, all remaining buyers also get to compete before buyer $i$. So, $p''$ completes the actual past scores $b^{<i}$ into a full threshold, as if buyer $i$ were last. To define it, for each $k\in\{i,\dots,n\}$, independently draw $T_k''\sim\mu_k$, set $b_k''=\Height(T_k'')$, and define
    \[
        p_j''
        :=
        \max\Bigl(
            \{b^{<i}_j\}
            \cup
            \{b_{k,j}'': i\le k\le n,\ j\in T_k''\}
        \Bigr)
        \qquad \forall j\in[m].
    \]
    Since $p_j''\ge b^{<i}_j$ for every item $j$, by the monotonicity of $v_{i}$, we have:
    \begin{align}
        \E[v_i(\ALG_i)|v_i=v_{i, \tau}, T_i=A] & \ge \E[v_{i, \type}(\{j \in A: b_{i,j} > p_j' \ge b^{<i}_j\})] \notag \\
                    & \ge \E[v_{i, \type}(\{j \in A: b_{i,j} > p_j' \ge p_j''\})] \label{eq: full-hallucinated-lb}
    \end{align}
    
    Now, combining the above with \Cref{lem:one-sixth-val}, we conclude that
    \[
        \E[v_i(\ALG_i) \mid v_i=v_{i, \tau}, T_i=A] \ge \frac{1}{6}v_{i, \type}(A).
    \]
\end{proof}

\emph{Remark.} In the proof above, the factor $6$ arises from comparing buyer $i$ to the
pessimistic threshold $p''$, which completes the actual past into a full
threshold as if buyer $i$ arrived last. This produces three i.i.d. threshold
vectors and hence the six weak orderings used in the final symmetry
argument of \Cref{lem:one-sixth-val}. In the i.i.d. arrival model, this last-buyer comparison is
pessimistic for buyers in the first half of the sequence, and is the key step of the analysis that is improved in \Cref{sec: iid}. 

\subsection{Sample-Access Variant}
\label{sec: infinite-type}

In this section, we show how \Cref{alg:online-algorithm} can be adapted when we only have
sampling access to each valuation distribution. The only object needed by the
$\Height{}$ procedure is an $\alpha$-subcritical collection of tentative-set
distributions. We define such distributions by sampling valuation profiles and
solving the corresponding offline LP.

For each buyer
$i\in[n]$, define a distribution $\mu_i$ over subsets of $[m]$ by \Cref{alg:infinite-type-mu}.

\begin{algorithm}[tbh]
\LinesNotNumbered
\caption{Sampler for $\mu_i$}
\label{alg:infinite-type-mu}
\begin{minipage}{0.96\textwidth}
\begin{enumerate}
    \item Sample $\widetilde v_k\sim\VD_k$ independently for all $k\in[n]$, and
    write $\widetilde{\mathbf v}=(\widetilde v_k)_{k\in[n]}$.
    \item Solve the offline version of \ref{lp:subadditive} on the valuation profile
    $\widetilde{\mathbf v}$, obtaining an optimal solution
    $\overline{x}^{\widetilde{\mathbf v}}$.
    \item Output a set $T$ by sampling from buyer $i$'s marginal in this LP
    solution: output $S\neq\emptyset$ with probability
    $\overline{x}^{\widetilde{\mathbf v}}_{i,S}$, and output $\emptyset$ with
    the remaining probability.
\end{enumerate}
\end{minipage}
\end{algorithm}

\begin{claim}
\label{clm:infinite-type-subcritical}
    The collection $(\mu_i)_{i\in[n]}$ is $\alpha$-subcritical.
\end{claim}
\begin{proof}[Proof of \Cref{clm:infinite-type-subcritical}]
    Fix an item $j\in[m]$. For each buyer $i$, by the definition of $\mu_i$,
    if $T_i\sim\mu_i$, then
    \[
        \Pr_{T_i\sim\mu_i}[j\in T_i]
        =
        \E_{\widetilde{\mathbf v}^{(i)}}
        \left[
            \sum_{S\ni j}
            \overline{x}^{\widetilde{\mathbf v}^{(i)}}_{i,S}
        \right],
    \]
    where $\widetilde{\mathbf v}^{(i)}\sim\prod_{k\in[n]}\VD_k$ is an auxiliary
    valuation profile used to define the marginal for buyer $i$.

    The profiles $\widetilde{\mathbf v}^{(i)}$ may be different for different
    buyers. However, since they all have the same distribution, linearity of
    expectation gives
    \[
        \sum_{i\in[n]}\Pr_{T_i\sim\mu_i}[j\in T_i]
        =
        \E_{\widetilde{\mathbf v}}
        \left[
            \sum_{i\in[n]}\sum_{S\ni j}
            \overline{x}^{\widetilde{\mathbf v}}_{i,S}
        \right],
    \]
    where $\widetilde{\mathbf v}\sim\prod_{k\in[n]}\VD_k$ is a single auxiliary
    profile.

    For every realized profile $\widetilde{\mathbf v}$, the solution
    $\overline{x}^{\widetilde{\mathbf v}}$ is feasible for the offline version
    of \ref{lp:subadditive}. Hence the item constraint for item $j$ gives
    \[
        \sum_{i\in[n]}\sum_{S\ni j}
        \overline{x}^{\widetilde{\mathbf v}}_{i,S}
        \le \alpha .
    \]
    Taking expectations yields
    \[
        \sum_{i\in[n]}\Pr_{T_i\sim\mu_i}[j\in T_i]\le \alpha .
    \]
    Since this holds for every $j\in[m]$, the collection
    $(\mu_i)_{i\in[n]}$ is $\alpha$-subcritical.
\end{proof}


The $\Height{}$ procedure can thus be defined using the distributions
$(\mu_i)_{i\in[n]}$. 
We now have the following natural adaptation of \Cref{alg:online-algorithm}, presented as \Cref{alg:sampling-online-algorithm}. Since the support of the valuation
distributions is no longer explicit, the arriving buyer's tentative set is
generated by completing its realized valuation to an independent sampled
valuation profile.

\begin{algorithm}[tbh]
\LinesNotNumbered
\caption{Online algorithm $\mathcal A'$ with sampling access}
\label{alg:sampling-online-algorithm}
\begin{minipage}{0.96\textwidth}
Let $\Height(\cdot)$ be the \SHSG{} procedure defined using the
induced distributions $\{\mu_i\}_{i\in[n]}$ from \Cref{alg:infinite-type-mu}.
\begin{enumerate}
    \item \textbf{Compute base thresholds.}
    Before any buyers arrive, independently sample $T_i'\sim\mu_i$ for each
    $i\in[n]$ using \Cref{alg:infinite-type-mu}, and set
    $b_i' := \Height{}(T_i')$.
    Define the threshold vector $p'$ coordinatewise by
    \[
        p_j'
        :=
        \max\Bigl(\{0\}\cup\{b_{i,j}' : i\in[n],\ j\in T_i'\}\Bigr)
        \qquad \forall j\in[m].
    \]

    \item \textbf{Online phase.}
    When buyer $i$ arrives and reveals valuation $v_i$:
    \begin{enumerate}
        \item Sample $\widetilde v_k^{(i)}\sim\VD_k$ independently for every
        $k\neq i$, and write
        $\widetilde{\mathbf v}_{-i}^{(i)}
        =(\widetilde v_k^{(i)})_{k\neq i}$. Solve the offline version of $LP(\alpha)$ on the completed
        valuation profile $\mathbf w^{(i)}:=(v_i,\widetilde{\mathbf v}_{-i}^{(i)})$, obtaining
        an optimal solution
        $\overline{x}^{\mathbf w^{(i)}}$.
        \item Sample a tentative set $T_i$ from buyer $i$'s marginal in this LP
        solution: output $S\neq\emptyset$ with probability
        $\overline{x}^{\mathbf w^{(i)}}_{i,S}$, and
        output $\emptyset$ with the remaining probability. Set
        $b_i:=\Height{}(T_i)$.
        \item Let $S_i:=\{j\in T_i: b_{i,j}>p_j'\}$.
        If $R_i$ is the set of items remaining when buyer $i$ arrives, allocate
        $\ALG_i:=S_i\cap R_i$ to buyer $i$.
    \end{enumerate}
\end{enumerate}
\end{minipage}
\end{algorithm}

We formally state its performance guarantee below.

\begin{theorem}
\label{thm:main}
    Given sampling access to the valuation distributions $\VD_i$ for all $i \in [n]$, and any $0 < \alpha < 1$, \Cref{alg:sampling-online-algorithm} produces a feasible allocation whose expected welfare is at least:
    \[  
        \frac{\alpha}{6}\E[\mathrm{OPT}]
    \]
    Moreover, the algorithm has expected running time polynomial in
    $n$, $m$, and $(1-\alpha)^{-1}$.
\end{theorem}

We obtain the following analogous conditional buyer-by-buyer guarantee as well.

\begin{lemma}
\label{lem:sampling-conditional-one-sixth}
    For \Cref{alg:sampling-online-algorithm}, for any buyer $i$, conditioned on the completed profile $\mathbf w^{(i)} = (v_{i},\widetilde{\mathbf v}_{-i}^{(i)})$, and the random tentative subset $T_i = A$, we have
    \[
        \E[v_i(\ALG_i) \; | \; \mathbf w^{(i)} = (v_{i},\widetilde{\mathbf v}_{-i}^{(i)}), T_i = A] \ge \frac{1}{6}v_{i}(A).
    \]
    In the above, the expectation is over the randomness in the types of the previous buyers, their random scores, and the threshold scores. 
\end{lemma}

With this lemma, we first complete the proof of \Cref{thm:main}. 

\begin{proof}[Proof of \Cref{thm:main}]
    By \Cref{clm:polytime}, sampling access to these
distributions is sufficient for the running-time guarantee of
\Cref{thm:bsg-guarantees}. Each sample from $\mu_i$ is implemented by
\Cref{alg:infinite-type-mu}, where in expectation, a polynomial number of offline instances of \ref{lp:subadditive} must be solved, which can be done in polynomial time by \Cref{clm: LP-OPT}. The expected running time then follows since the algorithm makes polynomially many calls to $\Height(\cdot)$. 

Now, we argue the competitive ratio. Conditioning on the completed profile $\mathbf w^{(i)} = (v_{i},\widetilde{\mathbf v}_{-i}^{(i)})$, and the random choice of the tentative subset $T_i = A$, by \Cref{lem:sampling-conditional-one-sixth}, we have
    \[
        \E[v_i(\ALG_i) \; | \; \mathbf w^{(i)} = (v_{i},\widetilde{\mathbf v}_{-i}^{(i)}), T_i = A] \ge \frac{1}{6}v_{i}(A).
    \]
    Averaging over the sampled tentative set, we get
    \[
        \E\!\left[
            v_i(\ALG_i)
            \,\middle|\,
            \mathbf w^{(i)} = (v_{i},\widetilde{\mathbf v}_{-i}^{(i)})
        \right]
        \ge
        \frac{1}{6}
        \sum_{S\subseteq[m]}
        \overline{x}^{\mathbf w^{(i)}}_{i,S} v_i(S).
    \]
    Taking expectations over $v_i$ and
    $\widetilde{\mathbf v}_{-i}^{(i)}$, and using that
    $\mathbf w^{(i)}$ has the same distribution as a fresh valuation profile
    $\mathbf v\sim\prod_{k\in[n]}\VD_k$, gives
    \[
        \E[v_i(\ALG_i)]
        \ge
        \frac{1}{6}
        \E_{\mathbf v}
        \left[
            \sum_{S\subseteq[m]}
            \overline{x}^{\mathbf v}_{i,S}v_i(S)
        \right].
    \]
    Summing over all buyers and using linearity of expectation,
    \[
        \E\!\left[\sum_{i\in[n]}v_i(\ALG_i)\right]
        \ge
        \frac{1}{6}
        \E_{\mathbf v}
        \left[
            \sum_{i\in[n]}\sum_{S\subseteq[m]}
            \overline{x}^{\mathbf v}_{i,S}v_i(S)
        \right].
    \]
    The expression inside the expectation is the optimal objective value of the
    offline version of \ref{lp:subadditive} on profile $\mathbf v$. For every realized
    profile $\mathbf v$, scaling the optimal integral allocation by $\alpha$ gives
    a feasible solution to \ref{lp:subadditive}, and hence this LP value is at least
    $\alpha\,\mathrm{OPT}(\mathbf v)$. Therefore,
    \[
        \E\!\left[\sum_{i\in[n]}v_i(\ALG_i)\right]
        \ge
        \frac{\alpha}{6}\,\E[\mathrm{OPT}].
    \]
\end{proof}
    
Choosing $\alpha=6/(6+\epsilon)$ gives the desired
$1/(6+\epsilon)$ fraction of the expected offline optimum, with only a
polynomial dependence on $1/\epsilon$ in the running time. 

 We now complete the proof of the conditional value-retention guarantee:

\begin{proof}[Proof of \Cref{lem:sampling-conditional-one-sixth}]
     We use the notation from the proof of \Cref{lem:conditional-one-sixth}. Condition on $\mathbf w^{(i)} = (v_{i},\widetilde{\mathbf v}_{-i}^{(i)})$ and $T_i=A$. Then, we have
    \[
        \ALG_i
        =
        \{j\in A: b_{i,j}>p_j'\ge b^{<i}_j\} \supseteq \{j\in A: b_{i,j}>p_j'\ge p_j''\}.
    \]
    Since $p_j''\ge b^{<i}_j$ for every item $j$, by the monotonicity of $v_{i}$, we have:
    \begin{align}
        \E[v_i(\ALG_i)|\mathbf w^{(i)} = (v_{i},\widetilde{\mathbf v}_{-i}^{(i)}), T_i=A] & \ge \E[v_{i}(\{j \in A: b_{i,j} > p_j' \ge p_j''\})] \label{eq: full-hallucinated-lb}
    \end{align}

    Combining the above with \Cref{lem:one-sixth-val}, 
    \[
        \E[v_i(\ALG_i)|\mathbf w^{(i)} = (v_{i},\widetilde{\mathbf v}_{-i}^{(i)}), T_i=A] \ge \frac{1}{6}v_{i}(A)
    \]
\end{proof}

%% file: iid.tex
\section{I.I.D. Arrival Model}

\label{sec: iid}


In this section, we refine the analysis of \Cref{alg:online-algorithm} in the
i.i.d. setting. We first prove the improved upper bound in
\Cref{sec:iid-ub}, and then show in \Cref{sec:ordering-barrier} that the
underlying ordering-cover argument cannot be pushed further within the same
analytic framework.

The improvement comes from revisiting the ``last-buyer'' comparison used in
\Cref{lem:one-sixth-val}. For buyers in the first half of the sequence,
completing the actual past into a full threshold is unnecessarily pessimistic:
one would like to compare such a buyer only to a midpoint completion of the
past. This is precisely where the i.i.d. assumption is used, since the buyers
before and after the midpoint have the same score distribution. Turning this
intuition into a value guarantee, however, requires an additional combinatorial
step. The midpoint comparison reduces the final value-retention argument to an
ordering-cover problem; a five-set cover then improves the retained-value factor
for first-half buyers from $1/6$ to $1/5$.

The same ordering-cover viewpoint also explains the limitation of our argument.
In \Cref{sec:ordering-barrier}, we abstract this final combinatorial step and
show that no smaller cover is possible for the natural generalization obtained
by splitting the buyers into more batches. Thus, the midpoint comparison gives
the best constant obtainable by refining this ordering-cover argument alone.

\subsection{Improved Upper Bound}
\label{sec:iid-ub}


For simplicity, assume that $n$ is even -- the improvement also holds for odd $n$ with a mild modification of the argument. Since the buyers are identically
distributed, we restrict attention to buyer-symmetric feasible
solutions\footnote{Any feasible solution can be symmetrized by averaging over
buyers; feasibility and objective value are preserved because the valuation
distributions are identical.} of \ref{lp:subadditive}. Thus we write
$\overline{x}_{\type,A}$ instead of $\overline{x}_{i,\type,A}$. The induced
distributions $\mu_i$ from \eqref{eq: mu_def} are then identical for all buyers;
we denote their common distribution by $\mu$. Consequently, over the randomness
of the buyers' valuations and tentative bundles, the score vectors generated by
different buyers are i.i.d.

We have the following guarantee for \Cref{alg:online-algorithm}. 

\begin{theorem}
\label{thm:iid}
    Given a buyer-symmetric feasible solution $\overline{x}$ to \ref{lp:subadditive},
    for some $0<\alpha<1$, \Cref{alg:online-algorithm} produces a feasible
    allocation whose expected welfare is at least
    \[
        \frac{11}{60}
        \left(
            n\sum_{\type\in\typeset}
            \sum_{A\neq\emptyset}
            \overline{x}_{\type,A}v_{\type}(A)
        \right).
    \]
\end{theorem}

We now give the improved analogue of \Cref{lem:one-sixth-val}. The proof proceeds by observing that the midpoint comparison reduces the value-retention step to five i.i.d.
half-threshold vectors: two for the threshold dominated by $b_i$, two for the
independent base threshold, and one for the midpoint-completed past.



\begin{lemma}
\label{lem:one-fifth-val}
    For \Cref{alg:online-algorithm} in the i.i.d. arrival model, fix any buyer
    $i\le n/2$. Then, for every type $\type\in\typeset$ and every tentative bundle
    $A\subseteq[m]$,
    \[
        \E\!\left[
            v_i(\ALG_i)
            \,\middle|\,
            v_i=v_{\type},\ T_i=A
        \right]
        \ge
        \frac{1}{5}v_{\type}(A).
    \]
    The expectation is over all remaining randomness of the algorithm, including
    the types and scores of earlier buyers and the threshold scores.
\end{lemma}

\begin{proof}
    We use the notation from the proof of \Cref{lem:conditional-one-sixth}. Condition on
    $v_i=v_{\type}$ and $T_i=A$. Then
    \[
        \ALG_i
        =
        \{j\in A: b_{i,j}>p_j'\ge b^{<i}_j\}.
    \]

    We complete the actual past only up to the midpoint of the sequence. For
    each $k\in\{i,\dots,n/2\}$, independently draw
    $T_k^{\mathrm{mid}}\sim\mu$, set
    $b_k^{\mathrm{mid}}=\Height(T_k^{\mathrm{mid}})$, and define
    \[
        p_j^{\mathrm{mid}}
        :=
        \max\Bigl(
            \{b^{<i}_j\}
            \cup
            \{b_{k,j}^{\mathrm{mid}}:
                i\le k\le n/2,\ j\in T_k^{\mathrm{mid}}\}
        \Bigr)
        \qquad \forall j\in[m].
    \]
    Since $p_j^{\mathrm{mid}}\ge b^{<i}_j$ for every $j$, the set
    $
        B^{\mathrm{mid}}
        :=
        \{j\in A: b_{i,j}>p_j'\ge p_j^{\mathrm{mid}}\}
    $
    satisfies
    \[
        B^{\mathrm{mid}}
        \subseteq
        \{j\in A: b_{i,j}>p_j'\ge b^{<i}_j\}
        =
        \ALG_i.
    \]

    We now use the stochastic-dominance guarantee. Let
    $\mathcal{F}_{\mathrm{half}}$ be the distribution of the coordinatewise
    maximum, restricted to $A$, of $n/2$ independent buyer-score vectors
    generated from tentative bundles sampled according to $\mu$. Draw
    $Q_1,Q_2$ independently from $\mathcal{F}_{\mathrm{half}}$. Since
    $\max\{Q_1,Q_2\}$ is distributed as a full threshold vector restricted to
    $A$, \Cref{thm:bsg-guarantees} gives $b_i=\Height(A)\succ \max\{Q_1,Q_2\}$.
    Hence there is a coupling $\omega$ of $b_i$ and $\max\{Q_1,Q_2\}$ such that $b_{i,j}>\max\{Q_{1j},Q_{2j}\}$ for every $j\in A$
    simultaneously.


    Next draw $Q_3,Q_4,Q_5$ independently from
    $\mathcal{F}_{\mathrm{half}}$. These vectors represent the thresholds appearing in
    $B^{\mathrm{mid}}$: since the buyers are i.i.d., $p'|_A$ is distributed as
    $\max\{Q_3,Q_4\}$, while by the construction of $p^{\mathrm{mid}}$,
    $p^{\mathrm{mid}}|_A$ is distributed as $Q_5$. Moreover, conditional on
    $v_i=v_{\type}$ and $T_i=A$, the vectors $b_i$, $p'$, and
    $p^{\mathrm{mid}}$ are mutually independent. Define the set 
    \[B_Q := \left\{
        j\in A:
        b_{i,j}>\max\{Q_{3j},Q_{4j}\}\ge Q_{5j}
    \right\}\]
    Using the monotonicity of $v_{\type}$ and $\ALG_i\supseteq B^{\mathrm{mid}}$, we have
    \begin{align}
    \label{eq: exp-eq}
        \E\!\left[
            v_i(\ALG_i)
            \,\middle|\,
            v_i=v_{\type},\ T_i=A
        \right] \ge \E[v_{\type}(B^{\mathrm{mid}})]
        =
        \E\!\left[
            v_{\type}\!\left(
                B_Q
            \right)
        \right],
    \end{align}
    where the rightmost expectation is taken under the coupling $\omega$ and
    over independent draws $Q_3,Q_4,Q_5$.

    Define
    \[
        \widetilde B
        :=
        \left\{
            j\in A:
            \max\{Q_{1j},Q_{2j}\}
            \ge
            \max\{Q_{3j},Q_{4j}\}
            \ge
            Q_{5j}
        \right\}.
    \]
    Under the coupling $\omega$, if $j\in\widetilde B$, then
    \[
        b_{i,j}>
        \max\{Q_{1j},Q_{2j}\}
        \ge
        \max\{Q_{3j},Q_{4j}\}
        \ge
        Q_{5j}.
    \]
    Hence, we have $\widetilde B \subseteq B_Q$. Using the monotonicity of $v_{\type}$ and Eq \eqref{eq: exp-eq}, we obtain
    \begin{align}
    \label{eq:b-tilde-lb}
        \E\!\left[
            v_i(\ALG_i)
            \,\middle|\,
            v_i=v_{\type},\ T_i=A
        \right]
        \ge
        \E[v_{\type}(B_Q)] \ge
        \E[v_{\type}(\widetilde B)].
    \end{align}
    where the rightmost expectation is taken over independent draws, $Q_1, \dots Q_5$.
    It remains to lower bound this final expectation. Label the vertices of a
    $5$-cycle in the order $1-5-2-3-4-1$.
    For $l\in[5]$, let $N(l)$ denote the two neighbors of $l$ in this cycle, and
    define
    \[
        R_l
        :=
        \Bigl\{
            j\in A:
            \max_{r\in N(l)} Q_{rj}
            \ge
            \max_{\substack{r\notin N(l)\\ r\neq l}} Q_{rj}
            \ge
            Q_{lj}
        \Bigr\}.
    \]
    With the chosen cyclic order, the set $\widetilde B$ is exactly $R_5$.

    We claim that the sets $R_1,\dots,R_5$ cover $A$ for every realization of
    $Q_1,\dots,Q_5$. Fix $j\in A$, and let $a$ attain
    $\max_{r\in[5]}Q_{rj}$. Let $b,c$ be the two neighbors of $a$ in the cycle.
    Since $a\in N(b)\cap N(c)$, the first inequality in the definitions of both
    $R_b$ and $R_c$ holds. If $j\notin R_b$, then the second inequality for
    $R_b$ fails, so $Q_{bj}>Q_{cj}$. If $j\notin R_c$, then similarly
    $Q_{cj}>Q_{bj}$. These two inequalities cannot both hold. Hence
    $j\in R_b\cup R_c$, proving the cover claim.

    By subadditivity,
    \[
        \sum_{l=1}^5 v_{\type}(R_l)\ge v_{\type}(A).
    \]
    Since $Q_1,\dots,Q_5$ are i.i.d., the random sets $R_1,\dots,R_5$ are
    identically distributed. Therefore, since $\widetilde B = R_5$, by Eq \eqref{eq:b-tilde-lb}
    \[
        \E\!\left[
            v_i(\ALG_i)
            \,\middle|\,
            v_i=v_{\type},\ T_i=A
        \right]
        \ge
        \E[v_{\type}(R_5)]
        =
        \frac{1}{5}
        \E\!\left[\sum_{l=1}^5 v_{\type}(R_l)\right]
        \ge
        \frac{1}{5}v_{\type}(A).
    \]
\end{proof}

\emph{Remark.} The same conditional value retention guarantee holds odd $n$ with a mild adaptation of the covering argument. For completeness, we include its proof in \Cref{sec: iid-odd-n}.

We now aggregate the first-half improvement with the general $1/6$ guarantee for
the remaining buyers.

\begin{proof}[Proof of \Cref{thm:iid}]
    Let $\gamma_i=1/5$ for $i\le n/2$ and $\gamma_i=1/6$ for $i>n/2$.
    By \Cref{lem:one-fifth-val} for first-half buyers and
    \Cref{lem:one-sixth-val} for the remaining buyers, taking expectations over the types and tentative subsets, we have
    \[
        \E[v_i(\ALG_i)]
        \ge
        \gamma_i
        \sum_{\type\in\typeset}
        \sum_{A\neq\emptyset}
        \overline{x}_{\type,A}v_{\type}(A)
        \qquad \forall i\in[n].
    \]
    Summing over buyers and using that $n$ is even,
    \[
        \E\!\left[\sum_{i=1}^n v_i(\ALG_i)\right]
        \ge
        \frac{n}{2}\left(\frac{1}{5}+\frac{1}{6}\right)
        \sum_{\type\in\typeset}
        \sum_{A\neq\emptyset}
        \overline{x}_{\type,A}v_{\type}(A)
        =
        \frac{11}{60}
        \left(
            n\sum_{\type\in\typeset}
            \sum_{A\neq\emptyset}
            \overline{x}_{\type,A}v_{\type}(A)
        \right).
    \]
\end{proof}

As before, given demand-oracle access, we may choose $\overline{x}$ so that the
LP objective is at least $\alpha\cdot\E[\mathrm{OPT}]$. Setting $\alpha=\frac{60/11}{60/11+\epsilon}$,
the algorithm is
$(60/11+\epsilon)$-competitive in the i.i.d. setting.

\subsection{Combinatorial Barrier for the Ordering Argument}
\label{sec:ordering-barrier}

At its heart, the proof of \Cref{lem:one-fifth-val} used a purely combinatorial fact: 
five ordering-defined sets cover the tentative bundle $A$, and subadditivity then converts this
cover into a factor-$5$ value loss. It is natural to ask whether this final
step can be improved by splitting the i.i.d. buyers into more batches and using
a more efficient ordering cover. In this section, we show formally that this combinatorial covering approach, by itself, cannot yield an improved competitive ratio.

We abstract here the general analytic setup. Suppose the buyers are split into $k$ equal batches,
each of size $n/k$, and consider a buyer whose position lies in the $\ell$-th
batch, where $0<\ell<k$. Let $\mathcal{F}_{n/k}$ denote the distribution of the
coordinatewise maximum, restricted to the buyer's tentative bundle $A$, of $n/k$
independent buyer-score vectors. The full threshold dominated by the buyer's
score can then be represented as the maximum of $k$ independent draws from
$\mathcal{F}_{n/k}$, the independent base threshold $p'$ can be represented as
the maximum of another $k$ independent draws from $\mathcal{F}_{n/k}$, and the
threshold obtained by completing the actual past up to the end of the $\ell$-th
batch can be represented as the maximum of $\ell$ independent draws from
$\mathcal{F}_{n/k}$. Thus, measured in batch-sized pieces, the three
threshold quantities consist of $k$, $k$, and $\ell$ independent pieces,
respectively.

Thus we are led to the following deterministic ordering problem. Write $x_1,\dots,x_{2k+\ell}$ for vectors in
$\mathbb{R}^A$. A permutation $\sigma$ partitions these vectors into three
blocks of sizes $k,k,\ell$. The set
\[
    S_{\sigma}
    :=
    \left\{
        j\in A:
        \max_{1\le r\le k} x_{\sigma(r),j}
        \ge
        \max_{k<r\le 2k} x_{\sigma(r),j}
        \ge
        \max_{2k<r\le 2k+\ell} x_{\sigma(r),j}
    \right\}
\]
is precisely the set of coordinates for which this ordering certifies the
desired comparison: the dominated full threshold is above the base threshold,
and the base threshold is above the completed-past threshold up to the $\ell$-th batch. If a family of
permutations $\Pi$ has the property that the sets $\{S_{\sigma}:\sigma\in\Pi\}$
cover $A$ for every choice of vectors, then the same symmetry-and-subadditivity
argument would give a value-retention factor $1/|\Pi|$ for buyers in the
$\ell$-th batch.

Concretely, \Cref{lem:one-fifth-val} is the case $(k,\ell)=(2,1)$. The two
vectors $Q_1,Q_2$ form the full threshold dominated by $b_i$, the two vectors
$Q_3,Q_4$ form the independent base threshold $p'$, and the single vector $Q_5$
forms the midpoint threshold $p^{\mathrm{mid}}$. The five-cycle construction
gives five permutations whose associated sets cover $A$, yielding the factor
$1/5$.

The following elementary lower bound shows that this template cannot improve
on the factor $5$.

\begin{lemma}
\label{lem:ordering-cover-lower-bound}
    Let $\Pi$ be a family of permutations such that the sets
    $\{S_{\sigma}:\sigma\in\Pi\}$ cover $A$ for every choice of vectors
    $x_1,\dots,x_{2k+\ell}$. If $q=|\Pi|$, then
    \[
        q\ge 4+\frac{2\ell}{k}.
    \]
    In particular, since $q$ is an integer and $\ell>0$, every such cover has
    size at least $5$.
\end{lemma}

\begin{proof}
    For a permutation $\sigma$, write
    \[
        F_\sigma:=\{t\in[2k+\ell]:\sigma^{-1}(t)\le k\},
        \qquad
        M_\sigma:=\{t\in[2k+\ell]:k<\sigma^{-1}(t)\le 2k\},
    \]
    and
    \[
        L_\sigma:=\{t\in[2k+\ell]:2k<\sigma^{-1}(t)\le 2k+\ell\}.
    \]
    Thus $F_\sigma,M_\sigma,L_\sigma$ are the first, middle, and final blocks
    induced by $\sigma$, of sizes $k,k,\ell$, respectively.

    Fix $t\in[2k+\ell]$. We first show that there must exist
    $\sigma\in\Pi$ such that $t\in F_\sigma$. Otherwise, we can choose a coordinate
    $j\in A$ and define the values at this coordinate by
    $x_{tj}=1$ and $x_{sj}=0$ for all $s\neq t$. Then, for any
    $\sigma\in\Pi$, since $t\notin F_\sigma$, either $t\in M_\sigma$ or
    $t\in L_\sigma$, and hence $j\notin S_\sigma$ for every $\sigma\in\Pi$, contradicting the
    assumption that the sets $\{S_\sigma:\sigma\in\Pi\}$ cover $A$.

    We next show that, in fact, there must be at least two permutations
    $\sigma\in\Pi$ satisfying $t\in F_\sigma$. Suppose instead that there is a
    unique permutation $\sigma_0\in\Pi$ with $t\in F_{\sigma_0}$. Since
    $|L_{\sigma_0}|=\ell>0$, choose some $s\in L_{\sigma_0}$. Choose a
    coordinate $j\in A$ and define
    \[
        x_{t,j}=2,\qquad x_{s,j}=1,\qquad
        x_{r,j}=0\quad\text{for all }r\notin\{s,t\}.
    \]
    For any $\sigma\neq\sigma_0$, we have $t\notin F_\sigma$ by uniqueness of
    $\sigma_0$. If $t\in M_\sigma$, then
    \[
        \max_{r\in F_\sigma}x_{r,j}\le 1
        <
        2
        =
        \max_{r\in M_\sigma}x_{r,j},
    \]
    so $j\notin S_\sigma$. If $t\in L_\sigma$, then
    \[
        \max_{r\in M_\sigma}x_{r,j}\le 1
        <
        2
        =
        \max_{r\in L_\sigma}x_{r,j},
    \]
    so again $j\notin S_\sigma$.

    It remains to check $\sigma_0$. Since $t\in F_{\sigma_0}$ and
    $s\in L_{\sigma_0}$, while all indices in $M_{\sigma_0}$ have value $0$ at
    coordinate $j$, we have
    \[
        \max_{r\in M_{\sigma_0}}x_{r,j}=0
        <
        1
        =
        \max_{r\in L_{\sigma_0}}x_{r,j}.
    \]
    Thus $j\notin S_{\sigma_0}$ as well. This contradicts the cover
    assumption. Therefore, for every $t\in[2k+\ell]$, there are at least two
    permutations $\sigma\in\Pi$ such that $t\in F_\sigma$.

    Counting the pairs $(t,\sigma)$ with $t\in F_\sigma$ in two ways gives
    \[
        kq
        =
        \sum_{\sigma\in\Pi} |F_\sigma|
        \ge
        2(2k+\ell)
        =
        4k+2\ell.
    \]
    Dividing by $k$ yields $q\ge 4+2\ell/k$.
\end{proof}

Thus the factor $5$ is optimal for this particular ordering-cover approach,
although it does not rule out a different analysis that uses additional
structure beyond symmetry and subadditivity.

%% file: fixed_point.tex
\section{Existence of fixed points}\label{sec:fixed-point-intgap}


In this section, we will show that the mirroring protocol described in the technical overview (\Cref{sec:techniques}) can be adapted to prove the existence of a rounding scheme for the configuration LP: any fractional solution $x$ of the configuration LP can be rounded to an integral allocation, while losing a factor of at most 2 in its social welfare. This proves that the integrality gap of the configuration LP is at most 2 when the valuations are subadditive.

At a high level, the idea is as follows: We consider the class of allocations described in \Cref{sec:techniques}, where a tentative set $T_i \sim (x_{i,S})_{S \subseteq [m]}$ is sampled for every $i$, and a score vector $b_i \sim \psi _i(T_i)$ is drawn. Every buyer is then assigned the items in their tentative set for which their score is the highest amongst all competing buyers.
For any buyer $i$, and for a set of BSGs $\psi _{-i}:=(\psi _k)_{k\neq i}$ of the remaining buyers, we describe a \emph{mirroring protocol} $\mathcal{M}_i(x;\psi _{-i})$ with the property that, if buyer $i$ assigns scores according to the mirroring protocol and the other buyers assign scores according to their BSGs, then buyer $i$ gets assigned a set that has value at least $\frac{1-\epsilon }{2} \sum_{S}x_{i,S}v_i(S)$ in expectation.
Then, using a fixed point argument, we prove the existence of BSGs $(\psi^*_1,\ldots,\psi^*_n)$ such that $\psi^*_i = \mathcal{M}_i(x;\psi^*_{-i})$ holds simultaneously for all $i$.
Using the BSGs $(\psi ^*_1,\ldots,\psi ^*_n)$ to construct the allocation then guarantees that every player $i$ achieves value at least $\frac{1-\epsilon }{2} \sum_{S}x_{i,S}v_i(S)$ in expectation. This directly implies that the social welfare produced by the allocation is $\frac{1-\epsilon }{2}$ times the LP value, proving the desired integrality gap bound.

Note that this result is already implied by \cite{feige2006maximizing}, which directly gives polynomial time rounding with the same guarantee.
Nonetheless, this fixed point perspective is still interesting for the following reasons: 
Firstly, the proof strategy used here is very clean, in the sense that we only need to design the mirroring protocol to satisfy a `local' property (namely, that a single player $i$ gets good value), and the fixed point theorem does the heavy lifting of converting this local property to a global property (namely, that the resulting allocation has good social welfare). 
This can be also be thought of as a game, where a buyer $i$ responds to the strategies $(\psi _k)_{k\neq i}$ of the other buyers by `playing' the mirroring strategy $\mathcal{M}_i(x;\psi _{-i})$. The fixed point argument shows the existence of an `equilibrium' where all the buyers are mirroring each other simultaneously.
Secondly, it highlights that the recursive procedure in \Cref{sec:recHeightBSG} can be viewed as a sampling algorithm for a set of BSGs that come from a fixed point.
In fact, the {\Height} algorithm can be viewed as a recursive unrolling of this `equilibrium' fixed point where all buyers are mirroring one another.
This explains its applicability for the online setting, given that Correa-Cristi~\cite{correa2023constant} prove their existential result through a fixed point theorem. Further, given the connection to Feige's original algorithm that we explain in \Cref{sec:feige}, this also provides an alternative fixed-point interpretation of Feige's guiding graph algorithm.
The generality of this fixed point view is highlighted further in \Cref{sec:MMS}, where we provide an application of the same proof strategy for the Maximin share problem.

Now let us formally describe the mirroring protocol.
Let $x$ be a feasible solution to the configuration LP. We wish to round $x$ into an integral solution, while losing at most 1/2 of its social welfare.
We will slightly deviate from the mirroring protocol described in \Cref{sec:techniques} and instead work with the space of scores whose entries lie in $\{0,\ldots,L\}$, for some large $L$. This will ensure that the space of scores is finite, which will help in proving the fixed point existence. 

Fix $0<\epsilon<1$. For every $i \in [n],j \in [m]$, define $\kappa_{ij}:=\sum_{S\ni j}x_{i,S}$. Let $B=1/\epsilon$, and $L = \left\lceil\frac{2mB}{\epsilon}\right\rceil$.
Let $\Delta$ denote the space of distributions over vectors in $\{0,\ldots,L\}^m$. Recall that a BSG is a randomized map from a subset $A \subseteq [m]$ to a score vector in $\{0,\ldots,L\}^m$ supported on $A$. In other words, a BSG can be seen as a function $\psi : 2^{[m]}\to \Delta$. We will additionally consider only BSGs $\psi _i$ that satisfy the item-wise expectation bound $\mathbb{E}_{b_i \sim \psi _i(S)}[b_{i,j}] \le B$ for every $i \in [n], j \in [m]$ and $S \subseteq [m]$.
Let $\Lambda = \{\psi \text{ a BSG}: \mathbb{E}_{b \sim \psi(S)}[b_{j}] \le B, \text{ for all } j\in [m],S \subseteq [m]\}$ be this space of BSGs. 

Given BSGs $(\psi _1,\ldots,\psi _n) \in \Lambda^n$, the rounding algorithm proceeds by independently sampling a tentative set $T_i \sim (x_{i,S})_{S \subseteq [m]}$ for every buyer $i$, and then independently drawing scores $b_i \sim \psi _i(T_i)$ for every $i \in [n]$. Buyer $i$ is then assigned the set \[R^{\vec{\psi }}_i=\{j \in T_i: b_{i,j}>\max_{k \neq i} b_{k,j}\}.\]
That is to say, an item goes to the buyer who has the highest score for it.\footnote{This can also be seen as a simultaneous auction where buyers submit bids to the auctioneer that are drawn from their respective BSGs, and an item goes to the buyer with the highest bid on it.}
We drop the $\vec{\psi }$ in the superscript when the set of BSGs used to define $R_i^{\vec{\psi }}$ is clear from context.
First, observe that the sets $(R_1,\ldots,R_n)$ are pairwise disjoint. Hence $(R_1,\ldots,R_n)$ forms a valid integral allocation, regardless of the choice of BSGs $(\psi _1,\ldots,\psi _n)$. 


Fix a buyer $i \in [n]$. 
Given $\psi _{-i} \in \Lambda^{n-1}$, the \emph{mirroring protocol} $\mathcal{M}_i(x;\psi _{-i})$ is defined to be a BSG that, when evaluated on a subset $A \subseteq [m]$, outputs the distribution over scores $b_i$ returned by the following process:\footnote{We use $\psi _{-i}$ as a shorthand for $(\psi _k)_{k\neq i}$.}
\begin{itemize}
	\item Independently draw a (simulated) tentative set $T_k' \sim (x_{k,S})_{S \subseteq [m]}$ for each $k\neq i$, and independently draw $b'_k \sim \psi _k(T'_k)$. 
	\item If $b'_{k,j}=L$ for any $k \neq i, j \in T'_k$, define $b_i = \vec{0}$, the all-zero vector.
	\item Otherwise, define $b_i$ according to \[ b_{i,j} = 
	\begin{cases}
		1+\max_{k\neq i}b'_{k,j} & \text{ if } j \in A \\ 0 & \text{ otherwise}
	\end{cases}
	\] 
\end{itemize}

We take the maximum over an empty set of scores to be zero.
The following lemma shows that, under suitable expected-score bounds, if $\psi _i$ is obtained from mirroring $\psi _{-i}$, then buyer $i$ gets at least a $(1-\epsilon)/2$ fraction of his expected tentative set value in expectation.\footnote{This mirroring protocol for simultaneous auctions is based on a similar mirroring idea that appears in ~\cite{feldman2013simultaneous}. Our main contribution is that mirroring can be obtained for all buyers simultaneously.}

\begin{lemma}\label{lem:fp_mirror_guarantee}
	Fix $i \in [n]$ and BSGs $(\psi_k)_{k\neq i} \in \Lambda^{n-1}$.
	Suppose that $\psi _i = \mathcal{M}_i(x;\psi _{-i})$ was constructed from $(\psi _k)_{k\neq i}$ by the mirroring protocol. Then \[ \mathbb{E}[v_i(R_i)] \ge  \left(\frac{1-\epsilon }{2}\right) \sum_{S \subseteq [m]}x_{i,S}v_i(S) \] 
\end{lemma}
\begin{proof}
	Fix a set $A \subseteq [m]$ with $x_{i,A}>0$ and condition on $T_i=A$.
	Let $\Gamma$ be the event that $b'_{k,j}=L$ for some $k\neq i$ and $j\in[m]$. 
	Note that $\Gamma$ is independent of $T_i$.
	For every $k\neq i$, since $\psi _k \in \Lambda$, we have $\mathbb{E}_{b'_k\sim\psi_k(S)}[b'_{k,j}]\le B$ for every fixed input $S\subseteq[m]$ and every $j\in[m]$. Averaging over $T'_k$, and using that scores vanish outside $T'_k$, we obtain
	\[
		\mathbb{E}[b'_{k,j}]
		=\sum_{S\ni j}x_{k,S}\mathbb{E}_{b'_k\sim\psi_k(S)}[b'_{k,j}]
		\le B\sum_{S\ni j}x_{k,S}
		=\kappa_{kj}B,
	\]
	where the expectation on the left includes both the tentative set $T'_k$ and the scores drawn from $\psi_k(T'_k)$.
	Consequently, by the union bound, Markov's inequality, and the item constraints $\sum_{k\in[n]}\kappa_{kj}\le 1$,
	\begin{align*}
		\Pr[\Gamma\mid T_i=A]
		=\Pr[\Gamma]
		\le \sum_{j \in [m]} \sum_{k\neq i} \Pr[b'_{kj}\ge  L]
		\le\frac{1}{L}\sum_{j\in[m]}\sum_{k\neq i}\mathbb{E}[b'_{k,j}]
		\le\frac{B}{L}\sum_{j\in[m]}\sum_{k\neq i}\kappa_{kj}
		\le\frac{mB}{L}\le\frac{\epsilon}{2}.
	\end{align*}
	For every $j \in A$, let
	\[
		p_j := \max_{k\neq i}b_{k,j},
		\qquad
		p'_j := \max_{k\neq i}b'_{k,j},
	\]
	where $b'_k$ are the simulated scores used to construct $b_i$.
	Conditional only on $T_i=A$, the vectors $p=(p_j)_{j\in A}$ and $p'=(p'_j)_{j\in A}$ are independent and identically distributed: both are obtained by independently sampling tentative sets from $(x_{k,S})_{S\subseteq[m]}$ for $k\neq i$, applying the corresponding BSGs, and taking the pointwise maximum. These draws are independent of $T_i$.
	Define
	\[
		A_0 := \{j\in A:p'_j\ge p_j\},
		\qquad
		A_1 := \{j\in A:p_j\ge p'_j\}.
	\]
	Since $p,p'$ are i.i.d., the sets $A_0$ and $A_1$ are identically distributed conditional on $T_i=A$. Moreover, $A_0\cup A_1=A$. Thus, by subadditivity,
	\[
		2\mathbb{E}[v_i(A_0)\mid T_i=A]
		=\mathbb{E}[v_i(A_0)+v_i(A_1)\mid T_i=A]
		\ge v_i(A).
	\]
	Let $\overline{\Gamma}$ denote the complement of $\Gamma$.
	On the event $\overline{\Gamma}$, the mirroring construction gives $b_{i,j}=1+p'_j$ for every $j\in A$. Since all scores are integers, on this event we have
	\[
		R_i=\{j\in A:1+p'_j>p_j\}=A_0.
	\]
	On the event $\Gamma$, the construction sets $b_i=\vec{0}$, so $R_i=\emptyset$ and $v_i(R_i)=0$. Thus $\mathbb{E}[v_i(R_i) \mid T_i=A] = \Pr[\overline{\Gamma}]\mathbb{E}[v_i(A_0)  \mid T_i=A,\overline{\Gamma}]$.
	We also have that $\mathbb{E}[v_i(A_0) \mid T_i=A] = \Pr[\overline{\Gamma}] \mathbb{E}[v_i(A_0)  \mid T_i=A,\overline{\Gamma}] + \Pr[\Gamma] \mathbb{E}[v_i(A_0)  \mid T_i=A,\Gamma]$.
	Combining these equalities, we get,
	\begin{align*}
		\mathbb{E}[v_i(R_i)\mid T_i=A]
		=\mathbb{E}[v_i(A_0)\mid T_i=A]-\Pr[\Gamma]\mathbb{E}[v_i(A_0)  \mid T_i=A,\Gamma]
		\ge\frac{1-\epsilon}{2}v_i(A).
	\end{align*}
	Finally, taking expectations over the tentative set $T_i$ yields
	\[
		\mathbb{E}[v_i(R_i)]
		=\sum_{\substack{A\subseteq[m]\\x_{i,A}>0}}x_{i,A}\mathbb{E}[v_i(R_i)\mid T_i=A]
		\ge \frac{1-\epsilon}{2}\sum_{A\subseteq[m]}x_{i,A}v_i(A).
	\]
\end{proof}
 
Given the above, it is natural to consider a family of BSGs $\psi _1,\ldots,\psi _n$ where each $\psi _i$ is simultaneously obtained from mirroring the other BSGs $\psi _{-i}$. By the above lemma, using such a family of BSGs directly implies that every player gets $(1-\epsilon)/2$ their LP value in expectation, and thus the entire allocation has a $(1-\epsilon)/2$-approximate guarantee on its social welfare.
This is formalized in the definition and corollary below.

\begin{definition}\label{defn:mutually_mirroring}
	Let $x$ be a solution to the configuration LP.
	A family of BSGs $(\psi _i)_{i \in [n]}$ is said to be \emph{mutually mirroring} (with respect to $x$) if $\psi _i = \mathcal{M}_i(x;\psi _{-i})$ for all $i \in [n]$.
\end{definition}

\begin{corollary}\label{cor:fp_sw_guarantee}
	Suppose that $(\psi _i)_{i \in [n]}$ are mutually mirroring with respect to a configuration LP solution $x$. Then \[ \mathbb{E}\left[\sum_{i \in [n]}v_i(R_i)\right] \ge  \frac{1-\epsilon }{2} \sum_{i \in [n]}\sum_{S \subseteq [m]}x_{i,S}v_i(S) \] 
\end{corollary}
\begin{proof}
	By \Cref{defn:mutually_mirroring}, each $\psi _i$ is obtained from the mirroring protocol applied to $\psi _{-i}$. By \Cref{lem:fp_mirror_guarantee}, for each $i \in [n]$, we have $\mathbb{E}[v_i(R_i)] \ge  \frac{1-\epsilon }{2}\sum_{S \subseteq [m]}x_{i,S}v_i(S)$.
	Summing this inequality over all players gives the social welfare guarantee, $\mathbb{E}[\sum_{i \in [n]}v_i(R_i)]\ge  \frac{1-\epsilon }{2} \sum_{i \in [n]}\sum_{S \subseteq [m]}x_{i,S}v_i(S)$.
\end{proof}

We will show that the existence of a mutually mirroring family of BSGs can be shown directly via a fixed point argument. By the above corollary, this immediately implies the existence of a rounding scheme that rounds a fractional solution $x$ into an integral allocation that loses at most a factor of 2 in its social welfare. In other words, this shows that the integrality gap of the configuration LP is at most 2 when the valuations are subadditive.

\begin{lemma}\label{lem:fp_existence}
	Let $x$ be a feasible solution to the configuration LP such that $\sum_{i \in [n]}\sum_{S\ni j}x_{i,S}\le 1-\epsilon$ for all $j \in [m]$.
	Then there exists a family $(\psi_i)_{i\in[n]}\in\Lambda^n$ of BSGs that are mutually mirroring with respect to $x$.
\end{lemma}
\begin{proof}
	We use Brouwer's fixed-point theorem, which states that if $K$ is a non-empty, convex, compact subset of $\mathbb{R}^D$ and $f:K\to K$ is continuous, then $f$ has a fixed point $x^*$ such that $x^*=f(x^*)$.
	By the assumption, $\sum_{i\in[n]}\kappa_{ij}\le 1-\epsilon$ for every $j\in[m]$, and $B\epsilon=1$.

	Set $E:=\{0,\ldots,L\}^m$. First, let us define the set $K$.
	For a family of BSGs $(\psi _1,\ldots,\psi _n)$, we represent them by the array $(z_{i,S,b})_{i \in[n],S \subseteq [m], b \in E}$, where $z_{i,S,b}:=\Pr_{b_i \sim \psi _i(S)}[b_i=b]$.
	There are $D:=n2^m(L+1)^m$ coordinates. Let $K\subseteq\mathbb{R}^D$ be the set of arrays satisfying, for every $i\in[n]$ and $S\subseteq[m]$,
	\begin{align*}
		\sum_{b\in E}z_{i,S,b}&=1,\\
		\sum_{b\in E}b_jz_{i,S,b}&\le B &&\text{for every }j\in[m],\\
		z_{i,S,b}&\ge 0 &&\text{for every }b\in E,\\
		z_{i,S,b}&=0 &&\text{if }b_j>0\text{ for some }j\notin S.
	\end{align*}
	Thus $K$ exactly represents the BSGs in $\Lambda^n$, where for every $i\in [n],S \subseteq [m]$, the probability distribution of $\psi _i(S)$ is specified in the coordinates of $z$.
	The set $K$ is nonempty, since the distributions concentrated on the all-zero vector satisfy all the constraints, and it is closed and convex, since every constraint is linear. Moreover, $K$ is also bounded, and hence compact.

	We now construct the function $f:K\to K$.
	For $i\in[n]$, $S\subseteq[m]$, and $b_{-i}=(b_k)_{k\neq i}\in E^{n-1}$, let $q_i^S(b_{-i})=\vec{0}$ if $b_{k,j}=L$ for some $k\neq i$ and $j\in[m]$. Otherwise, set
	\[
		q_{i,j}^S(b_{-i}):=
		\begin{cases}
			1+\max_{k\neq i}b_{k,j},&j\in S,\\
			0,&j\notin S.
		\end{cases}
	\]
	Given $z\in K$, independently sample $T'_k\sim(x_{k,T})_{T\subseteq[m]}$ for every $k\neq i$, and then independently draw $b'_k$ from the distribution with probability vector $(z_{k,T'_k,b})_{b\in E}$.
	Define
	\[
		f_{i,S,b}(z):=\Pr[q_i^S(b'_{-i})=b], \text{ and } f(z):=(f_{i,S,b}(z))_{i,S,b}.
	\]
	In other words, if $z$ represents the BSGs $(\psi _1,\ldots,\psi _n)$, then $f(z)$ represents $(\mathcal{M}_1(x;\psi _{-1}),\ldots,\mathcal{M}_n(x;\psi _{-n}))$.

	We verify that $f(z)\in  K$ for every $z \in K$.
	The output distributions satisfy the nonnegativity, probability, and support constraints by construction.
	For every $k\neq i$ and $j\in[m]$, the support and expected-score constraints give
	\[
		\mathbb{E}[b'_{k,j}]
		=\sum_{T\ni j}x_{k,T}\sum_{b\in E}b_jz_{k,T,b}
		\le B\sum_{T\ni j}x_{k,T}
		=B\kappa_{kj}.
	\]
	For every $j\in S$, returning the zero vector can only decrease the output score, so
	\begin{align*}
		\sum_{b\in E}b_jf_{i,S,b}(z)
		=\mathbb{E}[q_{i,j}^S(b'_{-i})]
		\le 1+\sum_{k\neq i}\mathbb{E}[b'_{k,j}]
		\le 1+B\sum_{k\neq i}\kappa_{kj}
		\le 1+B(1-\epsilon)=B.
	\end{align*}
	The last equality uses $B\epsilon=1$. For $j\notin S$, the output score is zero. Hence all expected-score constraints hold, proving that $f$ maps $K$ to itself.

	Finally, $f$ is continuous. Indeed, for every $i,S,b$, the exact formula that describes $f$ is
	\[
		f_{i,S,b}(z)
		=\sum_{(T_k)_{k\neq i}\in(2^{[m]})^{n-1}}
		\ \sum_{\substack{(b_k)_{k\neq i}\in E^{n-1}\\q_i^S(b_{-i})=b}}
		\prod_{k\neq i}x_{k,T_k}z_{k,T_k,b_k}.
	\]
	Thus, every coordinate of $f$ is a polynomial in $z$. This proves continuity.

	Brouwer's theorem now gives a point $z^*\in K$ with $f(z^*)=z^*$.
	For every $i\in[n]$ and $S\subseteq[m]$, let $\psi^*_i(S)$ be the distribution with probability vector $(z^*_{i,S,b})_{b\in E}$.
	Then $\psi_i^*\in\Lambda$, and the fixed-point equality gives
	\[
		\mathcal{M}_i(x;\psi^*_{-i})(S)\overset{d}{=}\psi^*_i(S)
		\qquad\text{for every }i\in[n],\ S\subseteq[m].
	\]
	Thus these BSGs are mutually mirroring with respect to $x$, as required.
\end{proof}

\begin{corollary}
	The integrality gap of the configuration LP is at least 1/2 when the valuations $v_1,\ldots,v_n$ are subadditive.
\end{corollary}
\begin{proof}
	Let $x$ be the optimal solution to the configuration LP with valuations $v_1,\ldots,v_n$. The objective value of this solution is $\sum_{i \in [n]}\sum_{S \subseteq [m]}x_{i,S}v_i(S)$. Let $x' = (1-\epsilon )x$. So we have $\sum_{i \in [n]}\sum_{S\ni j}x'_{i,S}\le 1-\epsilon $ for every $j \in [m]$.
	By \Cref{lem:fp_existence}, there exists $\psi _1,\ldots,\psi _n$ that are mutually mirroring with respect to $x'$.
	Let $\vec{\psi }=(\psi _1,\ldots,\psi _n)$, and consider the allocation $(R_1^{\vec{\psi }},\ldots,R_n^{\vec{\psi }})$.
	By \Cref{cor:fp_sw_guarantee}, this allocation guarantees $\mathbb{E}[\sum_{i \in [n]}v_i(R_i^{\vec{\psi }})] \ge  \frac{1-\epsilon }{2} \sum_{i \in [n]}\sum_{S \subseteq [m]}x'_{i,S}v_i(S) = \frac{(1-\epsilon )^2}{2} \sum_{i \in [n]}\sum_{S \subseteq [m]}x_{i,S}v_i(S)$. Sending $\epsilon \to 0$, this proves the existence of an allocation whose social welfare is at least 1/2 of the LP value, proving the corollary.
\end{proof}

\textbf{Remarks.} We remark that the proof of the integrality gap result can be done more directly via a Nash equilibrium argument. Consider a game where each buyer $i\in [n]$ plays a BSG $\psi _i \in \Lambda $. Buyer $i$'s utility under strategy profile $(\psi _1,\ldots,\psi _n)$ is defined to be $\mathbb{E}[v_i(R_i^{\vec{\psi }})]$.
This game has a (pure) Nash equilibrium since the strategy space $\Lambda $ of each buyer is convex, closed, bounded, and since the utilities are continuous in the probability-vector representation of $\Lambda^n$.
At a Nash equilibrium $(\overline{\psi }_1,\ldots,\overline{\psi }_n)$ of this game, the expected social welfare generated by $(R_1,\ldots,R_n)$ is at most a factor 2 from the optimal configuration LP value. To see this, suppose that some buyer $i \in [n]$ gets utility strictly less than $\frac{1-\epsilon }{2}\sum_{S}x_{i,S}v_i(S)$. Then this buyer may deviate to playing the mirroring strategy $\mathcal{M}_i(x;\overline{\psi } _{-i})$, which would strictly improve the buyer's utility by \Cref{lem:fp_mirror_guarantee}.\footnote{The deviation $\mathcal{M}_i(x;\overline{\psi }_{-i})$ will be a feasible strategy provided that there is slack in the configuration LP constraints.}
This is a contradiction to $(\overline{\psi }_1,\ldots,\overline{\psi }_n)$ being a Nash equilibrium.
Thus, every buyer must earn value at least $\frac{1-\epsilon }{2}\sum_{S}x_{i,S}v_i(S)$ in expectation at the Nash equilibrium. Summing this over all the buyers give the social welfare guarantee.

The Nash equilibrium may not be computable in polynomial time since the best response for a buyer $i$, given the BSGs $\psi _{-i}$ of the other buyers, may be hard to analyze.
The mutually mirroring equilibrium instead restricts every buyer's response to be the mirroring protocol, which is easily computable. Unrolling this mutually mirroring fixed point (whose existence is proved in \Cref{lem:fp_existence}) gives us the polynomial time algorithm that computes\footnote{Here by `computes', we mean that we get a polynomial time algorithm that returns a sample from $\psi ^*_i(S)$, for any $i \in [n], S \subseteq [m]$. Sample access suffices for us to compute the final allocation $(R_1,\ldots,R_n)$.} the mutually mirroring BSGs $(\psi _1^*,\ldots,\psi _n^*)$.
Now the allocation $R_1,\ldots,R_n$ obtained from $(\psi _1^*,\ldots,\psi ^*_n)$ would give us a polynomial time algorithm achieving an approximation factor of $(2+\epsilon )$ with respect to the configuration LP. This is precisely the idea behind \Cref{sec:offline_alg}.

%% file: MMS-appendix.tex
\section{Applications to Maximin Share}\label{sec:MMS}


\paragraph{Background in MMS.} 
In the maximin share problem, one is given a combinatorial auction instance over $m$ items and $n$ players having subadditive valuations $v_1,\ldots,v_n$ respectively. Let $\mathcal{P}$ be the space of all allocations (partitions) of the $m$ items to the $n$ players.
For $i \in [n]$, define the Maximin share: \[ \operatorname{MMS}_i^n = \max_{(S_1,\ldots,S_n) \in \mathcal{P}} \min_{k \in [n]}v_i(S_k) \] 
The goal is to return an allocation $(A_1,\ldots,A_n)$ such that $v_i(A_i)\ge c\cdot \operatorname{MMS}_i^n$ for all $i \in [n]$, for some desired approximation ratio $c$.

Currently, it is known that subadditive valuations admit $O(\log \log n)$-approximation MMS allocations~\cite{Feige2025FromMT,Seddighin2025BeatingTL}. The question of whether there exist $O(1)$-approximate MMS allocations for subadditive valuations is an open question.

The current best $O(\log \log n)$-MMS result uses a configuration LP rounding procedure developed by \cite{Seddighin2025BeatingTL}. Their rounding procedure is a modification of Feige's guiding graph algorithm, and provides a partial allocation to a subset of players, such that most of the players get a constant fraction of their MMS value with high probability.
Intuitively, Feige's original guiding graph algorithm only gives a guarantee on the expected value that a player gets in the returned allocation. However, there is no control on how well this value is concentrated around its expectation.
The paper \cite{Seddighin2025BeatingTL} modifies Feige's algorithm in a manner so that the resulting allocation is 'concentrated' in some manner, but only allocates items to a subset of players.
A rounding algorithm with such a guarantee can then be utilised to get an $O(\log \log n)$-approximate MMS allocation.

In this section, we will show that we can recast this modified rounding procedure as a fixed-point problem. Then Brouwer's theorem will tell us that the desired allocation exists.\footnote{Even proving the \emph{existence} of good MMS allocations is a non-trivial problem. We remark that, given the solution $x$, the allocation in \Cref{lem:MMS} can be found in polynomial time via a similar recursive procedure obtained from unrolling the fixed-point. However, we will not prove it here.}
We believe that this makes for a simpler proof of the result.

We will use the following fact that is true for MMS instances.
MMS instances for combinatorial auctions induce a special configuration LP solution $(x_1,\ldots,x_n)$ that has the properties:
\begin{itemize}
	\item for any $i \in [n]$ and $S \in \operatorname{supp}(x_i)$, we have $v_i(S)\ge 1$, and
	\item $\sum_{S:j \in S}x_{i,S}=\frac{1}{n}$ for every $j \in [m]$.
\end{itemize}

This solution $x$ is obtained as follows: First, scale the valuation of each player so that $\operatorname{MMS}_i^n=1$ for all $i \in [n]$. This can be done without loss of generality (see prior work on MMS for a proof~\cite{Seddighin2025BeatingTL}).
For each player $i \in [n]$, let $(S_1^i,\ldots,S^i_n)$ be their optimal MMS partition.
The fractional solution $x$ obtained by setting, for all $i \in [n]$, the values $x_{i,S}=\frac{1}{n}$ whenever $S=S^i_k$ for some $k \in [n]$, and setting $x_{i,S}=0$ otherwise, is a configuration LP solution that satisfies the above two properties.
Property (1) follows since $(S_1^i,\ldots,S_n^i)$ was an optimal MMS partition, and property (2) follows since $(S_1^i,\ldots,S_n^i)$ is a partition of $[m]$.

We will prove the following lemma that was shown in \cite{Seddighin2025BeatingTL}.

\begin{lemma}\label{lem:MMS}\cite{Seddighin2025BeatingTL}
	Let $[n],[m],v_1,\ldots,v_n$ be an MMS instance, normalized so that $\operatorname{MMS}^n_i=1$ for all $i \in [n]$.
	Let $d \in \mathbb{Z}_+$ and $Q \subseteq [n]$ with $|Q|\le \left\lfloor \frac{n}{d} \right\rfloor$.
	There exists $Q' \subseteq Q$ and an allocation $(A_i)_{i \in Q'}$ to the agents in $Q'$ such that
	\begin{itemize}
		\item $|Q'|\ge \frac{d}{d+1}|Q|$, and
		\item for every agent $i \in Q'$, we have $v_i(A_i)\ge \frac{1}{2}$.
	\end{itemize}
\end{lemma}
\begin{proof}
We prove this by a fixed point argument that goes along the lines of the argument in \cref{sec:fixed-point-intgap}.
Fix $0<\epsilon <1$.
Fix some $B \ge  \frac{n}{d}$, and let $L = \left\lceil \frac{2mB}{\epsilon } \right\rceil $.
We will discretize each bid to lie in $\{0,\ldots,L\}$. 
Thus, there are exactly $(L+1)^m$ bid vectors over the items. 
Let $x$ be a special solution to the configuration LP.

We will use the same mirroring ideas as before, but we wish to modify the mirroring protocol so that mutually mirroring distributions now give the type of guarantee stated in the lemma.
Roughly speaking, our goal is to get an allocation in which every player's value is more concentrated around its mean.
A natural idea for achieving this is to take multiple samples from each player's distribution when mirroring. 
We will show that modifying the mirroring protocol to take $d$ samples of each player's distribution, and then take the point-wise max of all sampled bids restricted on a sampled set, works out to give us exactly the guarantee in the lemma. 

Let $\Lambda $ be the space of distributions $\mathcal{B}$ over vectors $b \in \{0,\ldots,L\}^m$ satisfying $\mathbb{E}_{b\sim \mathcal{B}}[b_{j}]\le \frac{B}{n}$ for every $j \in [m]$.
Fix $i \in Q$ and distributions $\mathcal{B}_k \in \Lambda $ for each player $k \in Q-\{i\}$.
Define the modified mirroring protocol $\widehat{\mathcal{M}}_i(\mathcal{B}_{-i})$ to be the distribution over vectors $b_i$ obtained from the following process:
\begin{enumerate}
	\item Sample $b^1_k,\ldots,b^d_k \sim \mathcal{B}_k$ for every $k \in Q-\{i\}$ independently.
	\item If $b^\ell_{k j}=L$ for any $k \in Q-\{i\}$, $j \in [m]$, $\ell \in [d]$, then define $b_i$ to be the all-zero vector.
	\item Otherwise, sample a set $T \sim (x_{i,S})_{S \subseteq [m]}$ independently and define $b_i$ according to \[ b_{ij} =
	\begin{cases}
		1+ \max_{k \in Q-\{i\} } \max_{\ell\in [d]} b^\ell_{k j} & \text{ if } j \in T \\
		0 & \text{ otherwise}
	\end{cases} \] 
\end{enumerate}

Let's verify that $\widehat{\mathcal{M}}_i(\mathcal{B}_{-i}) \in \Lambda $.
Note that $\widehat{\mathcal{M}}_i(\mathcal{B}_{-i})$ is a distribution over vectors in $\{0,\ldots,L\}^m$ by construction.
We need to show that $\mathbb{E}_{b_i \sim \widehat{\mathcal{M}}_i(\mathcal{B}_{-i})}[b_{ij}] \le \frac{B}{n}$.
Fix some $j \in [m]$. We have,
\begin{align*}
	\mathbb{E}_{b_i \sim \widehat{\mathcal{M}}_i(\mathcal{B}_{-i})} \left[ b_{ij} \right] &\le \mathbb{E}_{T,b^1_k,\ldots,b^d_k \sim \mathcal{B}_k, \forall k\in Q-\{i\} } \left[ 1[j \in T] \cdot \left( 1+ \max_{k \in Q-\{i\}  } \max_{\ell \in [d]} b^\ell_{k j}\right) \right] \\ 
	&= \frac{1}{n} \cdot \mathbb{E}_{b^1_k,\ldots,b^d_k \sim \mathcal{B}_k, \forall k \in Q-\{i\} } \left[1+ \max_{k \in Q-\{i\}  }\max_{\ell \in [d]} b^\ell_{k j}\right]\\
	&\le \frac{1}{n}\cdot \mathbb{E}_{b^1_k,\ldots,b^d_k \sim \mathcal{B}_k, \forall k \in Q-\{i\} } \left[1+ \sum_{k \in Q-\{i\}  }\sum_{\ell \in [d]} b^\ell_{k j}\right]\\
	&= \frac{1}{n}+\frac{1}{n}\cdot \sum_{k \in Q-\{i\}  } \sum_{\ell\in [d]} \mathbb{E}_{b^\ell_k \sim \mathcal{B}_k}[b^\ell_{k j}] \\
	&\le  \frac{1}{n}+ \frac{1}{n} \cdot\sum_{k \in Q-\{i\}  } \sum_{\ell \in [d]} \frac{B}{n}\\
	&\le  \frac{1}{n}+ \frac{d}{n} \cdot \left(\left\lfloor \frac{n}{d} \right\rfloor - 1\right) \frac{B}{n}\\
	&\le  \frac{1}{n} + \frac{B}{n} - \frac{dB}{n^2}\\
	&\le \frac{B}{n} 
\end{align*}

The second line uses the independence of the sampled bids $b^1_{k},\ldots,b^d_k$ from the set $T$. The third upper bounds the max with a sum. The fifth line uses the fact that $\mathcal{B}_k \in \Lambda $ for all $k \in Q-\{i\}$.
The final inequality is due to $B \ge \frac{n}{d}$.


\begin{claim}\label{claim:mms_guarantee}
	Fix $0<\epsilon <1$, let \(L=\left\lceil 2mB/\epsilon\right\rceil\), and suppose that \((\mathcal B_k)_{k\in Q}\) is a fixed point of the truncated mirror map, that is, $\mathcal{B}_i\overset{d}{=}\widehat{\mathcal{M}}_i(\mathcal{B}_{-i})$. If the bid vectors are sampled independently and $A_i=\left\{j\in[m]:b_{ij}>\max_{k\in Q\setminus\{i\}}b_{kj}\right\}$, then, for every \(i\in Q\),
 \[
 \Pr\left[v_i(A_i)\ge \frac12\right]\ge (1-\epsilon)\frac{d}{d+1}.
 \]
\end{claim}
\begin{proof}
	Fix \(i\in Q\). Because \(\mathcal B_i=\widehat{\mathcal M}_i(\mathcal B_{-i})\), we may generate \(b_i\) as follows: Independently sample \(T_i\) according to \((x_{i,S})_{S\subseteq[m]}\), and, for every \(\ell\in[d]\) and \(k\in Q\setminus\{i\}\), independently sample
\(
b_k^\ell\sim\mathcal B_k.
\)
Independently of these samples, let \(b_k^0:=b_k\sim\mathcal B_k\) be the actual bid of each \(k\in Q\setminus\{i\}\). This coupling has the same distribution as independently sampling \(b_k\sim\mathcal B_k\) for every \(k\in Q\).

For every \(\ell\in\{0,\ldots,d\}\) and \(j\in[m]\), define
\(
Y_j^\ell:=\max_{k\in Q\setminus\{i\}}b_{kj}^\ell.
\)
Let \(\Gamma \) be the event that
\(
b_{kj}^\ell<L
\)
for every \(k\in Q\setminus\{i\}\), \(j\in[m]\), and \(\ell\in\{0,\ldots,d\}\).
Let \(\overline{\Gamma }\) denote the complement of \(\Gamma \).

We first bound the probability of \(\overline{\Gamma }\). For a fixed \(\ell\),
\[
\begin{aligned}
\Pr\left[\exists\,k\in Q\setminus\{i\},j\in[m]:b_{kj}^\ell=L\right]
&\le \sum_{k\in Q\setminus\{i\}}\sum_{j\in[m]}\Pr[b_{kj}^\ell=L]\\
&\le \frac1L\sum_{k\in Q\setminus\{i\}}\sum_{j\in[m]}\mathbb E[b_{kj}^\ell]\\
&\le \frac{m(|Q|-1)B}{nL}\\
&\le \frac{mB}{dL},
\end{aligned}
\]
where the second inequality is Markov's inequality, the third uses \(\mathcal B_k\in \Lambda \), and the last uses \(|Q|\le\lfloor n/d\rfloor\). Taking a union bound over \(\ell\in\{0,\ldots,d\}\) gives
\[
\Pr[\overline{\Gamma }]\le \frac{(d+1)mB}{dL}\le\frac{2mB}{L}\le\epsilon.
\]
Consequently,
\(
\Pr[\Gamma ]\ge1-\epsilon.
\)

Now fix \(X\in\operatorname{supp}(x_i)\) and condition on \(\Gamma \cap\{T_i=X\}\). Under this conditioning the mirror map does not return the zero vector, and hence
\[
b_{ij}=
\begin{cases}
1+\max_{\ell\in[d]}Y_j^\ell,&j\in X,\\
0,&j\notin X.
\end{cases}
\]
Because all bids are integral and nonnegative,
\[
A_i
=\left\{j\in X:\max_{\ell\in[d]}Y_j^\ell\ge Y_j^0\right\}.
\]

For every \(p\in\{0,\ldots,d\}\), define
\[
H_p:=\left\{j\in X:Y_j^p>\max_{\ell\in\{0,\ldots,d\}\setminus\{p\}}Y_j^\ell\right\},
\qquad \text{ and  }
W_p:=X\setminus H_p.
\]
Thus \(A_i=W_0\).

The sets \(H_0,\ldots,H_d\) are pairwise disjoint. They need not partition \(X\), because an item may have a tie for the largest bid. Nevertheless, for any distinct \(p,q\in \{0,\ldots,d\} \),
\[
W_p\cup W_q=(X\setminus H_p)\cup(X\setminus H_q)=X.
\]
By subadditivity, this implies that 
\(
v_i(X)\le v_i(W_p)+v_i(W_q).
\)
It follows that, for every realization of \((Y^0,\ldots,Y^d)\), there is at most one index \(p\) such that
\(
v_i(W_p)<\frac12v_i(X).
\)
Indeed, if this inequality held for two distinct indices \(p,q \in \{0,\ldots,d\} \), then we would have 
\(
v_i(X)\le v_i(W_p)+v_i(W_q)<v_i(X),
\)
a contradiction.

Conditional on \(\Gamma \), the vectors \(Y^0,\ldots,Y^d\) are exchangeable: before conditioning they are independent and identically distributed, and \(\Gamma \) is the intersection of the same coordinatewise event for each of the \(d+1\) samples. Therefore the random sets \(W_0,\ldots,W_d\) have identical marginal distributions under the conditioning.

Let
\(
Z_p:=\mathbf 1\left\{v_i(W_p)<\frac12v_i(X)\right\}.
\)
We have
\(
\sum_{p=0}^d Z_p\le1
\) for every realization.
Taking conditional expectations and using exchangeability gives
\[
(d+1)\Pr\left[v_i(W_0)<\frac12v_i(X)\,\middle|\,\Gamma ,T_i=X\right]
=\mathbb E\left[\sum_{p=0}^dZ_p\,\middle|\,\Gamma ,T_i=X\right]
\le1.
\]
Recall that $A_i=W_0$. Hence,
\[
\Pr\left[v_i(A_i)\ge\frac12v_i(X)\,\middle|\,\Gamma ,T_i=X\right]\ge\frac{d}{d+1}.
\]

By the first property of the special configuration LP solution \(x\), we have \(v_i(X)\ge 1\) for every \(X\in\operatorname{supp}(x_i)\). Averaging over \(T_i\) therefore yields
\[
\Pr\left[v_i(A_i)\ge\frac12\,\middle|\,\Gamma \right]\ge\frac{d}{d+1}.
\]
Finally,
\[
\begin{aligned}
\Pr\left[v_i(A_i)\ge\frac12\right]
&\ge
\Pr[\Gamma ]\Pr\left[v_i(A_i)\ge\frac12\,\middle|\,\Gamma \right]\\
&\ge
(1-\epsilon)\frac{d}{d+1}.
\end{aligned}
\]
\end{proof}

\begin{claim}\label{claim:mms_fixed_point_existence}
	There exists a set of distributions $(\mathcal{B}_i)_{i \in Q}$ over vectors in $\{0,\ldots,L\}^m$, such that $\mathcal{B}_i \overset{d}{=} \widehat{\mathcal{M}}_i(\mathcal{B}_{-i})$ for all $i \in Q$.
\end{claim}
\begin{proof}
	Let $E=\{0,\ldots,L\}^m$ be the space of bids. 
	We will first embed a distribution $\mathcal{B} \in \Lambda $ in a high-dimensional space $\mathbb{R}^D$ that has a coordinate for each bid vector $b \in E$.
	Thus $D=(L+1)^m$.
	For each $b \in E$, we set $z_b$ to be the probability that we obtain the vector $b$ when sampling from $\mathcal{B}$, and define $z=(z_b)_{b \in E}$ to be the representation of $\mathcal{B}$.
	Now we observe that every item budget constraint $\mathbb{E}_{b_i \sim  \mathcal{B}_i}[ b_{ij} ] \le \frac{B}{n}$ is a linear constraint in $z$.
	Thus, the set of distributions for agent $i$ can be represented as the set $P = \{z \in \mathbb{R}^D_+ : 1^\intercal z = 1, \sum_{b} b_{j}z_b \le \frac{B}{n}, \forall j \in [m] \} $, which is compact and convex. It is also non-empty as it contains the distribution concentrated on the all-zero vector.
	Let $K := P^{|Q|}$. Thus, $K$ denotes a family of distributions, one for each player in $Q$. Since $P$ is non-empty, convex and compact, it follows that $K$ is also non-empty, convex, and compact.
	
	Now we define the function $f:K \to K$ on which we will apply Brouwer's theorem.
	Given a point $z=(z_i)_{i \in Q} \in K$, let $\mathcal{B}_i$ be the distribution represented by $z_i$, for every $i \in Q$. 
	We define $f(z)$ to be the representation of the distributions $(\widehat{\mathcal{M}}_i(\mathcal{B}_{-i}))_{i \in Q}$.
	Note that $f(z) \in K$ because $\widehat{\mathcal{M}}_i(\mathcal{B}_{-i})$ lies in $\Lambda $ for every $i \in Q$.
	
	Now let us show that $f$ is continuous.
	Given $i \in Q$, $S \subseteq [m]$, and bids $b^1_k,\ldots,b^d_k \in \{0,\ldots,L\}^m$ for every $k \in Q-\{i\}$, we define $q^S_i(b^1_{-i},\ldots,b^d_{-i}) \in \{0,\ldots,L\}^m$ as the vector that is the all-zero vector if $b^\ell_{k j}=L$ for some $k\in Q-\{i\}, \ell\in [d], j \in [m]$, and is given by 
	\[ q_{ij}^S(b^1_{-i},\ldots,b^d_{-i}) = 
	\begin{cases}
		1+ \max_{k \in Q-\{i\} } \max_{\ell \in [d]} b^\ell_{k j} & \text{ if } j \in S \\
		0 & \text{ otherwise}
	\end{cases}
	\]
	if $b^\ell_{k j}\le L-1$ for all $k \in Q-\{i\}$, $j \in [m]$, and $\ell\in [d]$. 
	In other words, $q_i^S(b^1_{-i},\ldots,b^d_{-i})$ is the vector produced by player $i$'s mirroring protocol when the sampled set is $S$, and the samples of the other players are $b^1_{-i},\ldots,b^d_{-i}$. 
	Given this, one can see that the coordinates of $f(z)$ are exactly given by,
	\[f_{i,q}(z) = \sum_{S \subseteq [m]} x_{i,S} \sum_{\substack{ (b^\ell_k)_{k \in Q-\{i\}, \ell\in [d] } \\ \text{ s.t. }q = q_i^S (b^1_{-i},\ldots,b^d_{-i}) }}  \prod_{\substack{k \in Q-\{i\} \\ \ell\in [d] } } z_{k,b^\ell_k}  \]
	for every $i \in Q$ and $q \in \{0,\ldots,L\}^m$. 
	Thus, every coordinate of $f$ is a polynomial in its input $z$, and hence $f$ is continuous.

	By Brouwer's theorem, there exists a fixed point $f(z)=z$. This corresponds exactly to the claimed set of distributions.
\end{proof}

Now let us finish the proof of the lemma.
Fix a set of distributions $(\mathcal{B}_i)_{i \in Q}$ satisfying $\mathcal{B}_i\overset{d}{=}\widehat{\mathcal{M}}_i(\mathcal{B}_{-i})$ for all $i \in Q$, the existence of which is guaranteed by \Cref{claim:mms_fixed_point_existence}.
By \Cref{claim:mms_guarantee}, for each player $i \in Q$, we have that $v_i(A_i)\ge \frac{1}{2}$ with probability at least $(1-\epsilon )\frac{d}{d+1}$.
Thus, the expected number of agents for which $v_i(A_i)\ge \frac{1}{2}$ is at least $(1-\epsilon )\frac{d}{d+1}|Q|$.
It follows that there exists an allocation $(A_i)_{i \in Q}$ under which $v_i(A_i)\ge \frac{1}{2}$ holds for at least $(1-\epsilon )\frac{d}{d+1}|Q|$ agents in $Q$.
Taking $\epsilon \to 0$, and setting $Q'$ to be the corresponding set of agents gives the lemma.

\end{proof}

%% file: XOS-Appendix.tex
\section{XOS Valuations via Online Contention Resolution Schemes}
\label{app:xos-ocrs-reduction}
In this section, we consider the online combinatorial auction problem
with XOS valuations, a special case of subadditive valuations. As such, we again have items $[m]$,
yet we now consider when the valuations of the buyers are XOS, or equivalently fractionally subadditive. Formally,
$v:2^{[m]}\to \R_{\ge 0}$ is XOS, provided
that for every $S\subseteq[m]$, there is a non-negative vector $(a^{S}_{j})_{j\in[m]} \in \R_{\ge 0}^m$, referred to as an \textit{additive clause} for $S$, such that
\begin{equation}
\label{eq:xos-supporting-clause}
    v(S) = \sum_{j\in S} a^{S}_{j},
    \qquad
    v(B) \ge \sum_{j\in B} a^{S}_{j}
    \quad \text{for every } B\subseteq[m].
\end{equation}
The goal of the section is to reduce the design online algorithm which is $\gamma$-competitive, 
to an \textit{online contention resolution scheme} (OCRS) which is $\gamma$-selectable.
While the main focus of the paper was when each item can be assigned to at most one buyer (i.e., unit supply),  in this section we allow for a more general feasibility constraint on the items.  For each item $j \in [m]$, let $\mathcal F_j \subseteq 2^{[n]}$ be a downward-closed family describing which buyers may simultaneously receive item $j$.  Thus, an allocation $(A_i)_{i\in[n]}$ is feasible if
\[
    \{i\in[n] : j\in A_i\} \in \mathcal F_j
    \qquad \text{for every } j\in[m].
\]
For instance $\mathcal{F}_j$ could correspond to assigning $j$ to at most $k_j \ge 1$ buyers,
or more generally to a matroid or knapsack constraint.

Before stating our reduction, we provide the formal defintiion of an OCRS for an arbitrary constraint system $\mathcal{F}$ on ground set $[n]$. It will be convenient to denote
$\mathcal{P}$ as the convex hull of vectors in $\mathbb{R}^n$ corresponding
to feasible solutions of $\mathcal{F}$. We refer to $\mathcal{P}$
as the convex hull of $\mathcal{F}$ for brevity.
\begin{definition}
An OCRS for $\mathcal F$ takes as input
an arbitrary vector $y \in \mathcal{P}$. Each $i \in [n]$ has an independently drawn Bernoulli random variable $Y_i \sim \Ber(y_i)$, where $i$ is said to be \textit{active} if $Y_i =1$. Initially, the instantiations of $(Y_i)_{i \in [n]}$ are unknown to the OCRS.
In each step $t=1, \ldots , n$, it is revealed $Y_t \in \{0,1\}$. If $Y_t =0$, it cannot accept
$t$. Otherwise, if $Y_t = 1$, it makes an irrevocable decision on whether to accept $t$, subject to the feasibility constraint $\mathcal{F}$. If $A(y)$ denotes
the accepted elements of $[n]$, then the OCRS is said to be $\gamma$-selectable for $y$, provided
\begin{equation}
\label{eq:selectability}
    \Pr[\text{$i \in A(y)$}\mid Y_{i}=1]
    \ge \gamma
    \qquad \forall i\in[n].
\end{equation}
If \eqref{eq:selectability} holds for all $y \in \mathcal{P}$, then the OCRS is said to be
$\gamma$-selectable for $\mathcal{F}$.

\end{definition}
We are now ready to formalize the aforementioned reduction:
\begin{theorem}
\label{thm:xos-ocrs-reduction}
Suppose that for every item $j\in[m]$ there is a $\gamma_j$-selectable OCRS for $\mathcal F_j$.  If $\gamma = \min_{j\in[m]}\gamma_j$, then there is an $\gamma$-competitive algorithm for the online combinatorial allocation problem with XOS valuations.

\begin{remark}
Suppose that for each item $j$, the given OCRS is polynomial time, and $\mathcal{P}_j$, the convex hull of $\mathcal{F}_j$, can be described by polynomially many linear inequalities. Then, the resulting online algorithm in \Cref{thm:xos-ocrs-reduction} is polynomial time (assuming access to demand oracles). 
\end{remark}

\begin{remark}
If the buyers $[n]$ are assumed to arrive in a uniformly at random order, then the use of an OCRS for each item $j \in [m]$ can be replaced
with a \textit{random-order contention resolution scheme (RCRS)} for $\mathcal{F}_j$. In general, RCRS's yield better selection guarantees than OCRS's, and so the competitive guarantee of \Cref{thm:xos-ocrs-reduction} will improve.
\end{remark}

\end{theorem}


\subsection{Proof of \Cref{thm:xos-ocrs-reduction}}
We use the same notation as in \Cref{sec:preliminaries}. That is, $\Gamma_i$ denotes the type set of buyer $i \in [m]$, $\typeprob_{i,\tau}$ denotes the probability that buyer $i$ has type $\tau \in \Gamma_i$, and $v_{i,\tau}$ denotes the valuation of buyer $i$, assuming it has type $\tau$.
To prove \Cref{thm:xos-ocrs-reduction}, we provide a guarantee against the standard configuration LP on the input with XOS valuations. The variables of the below LP have the same interpretation as \ref{lp:subadditive}.

\begin{align} \label{lp:xos}
\tag{LP-XOS}
\max \quad
& \sum_{i\in[n]} \sum_{\tau\in\Gamma_i} \sum_{\emptyset\ne S \subseteq[m]} x_{i,\tau,S} \cdot v_{i,\tau}(S) \notag \\
\text{s.t.} \quad
& \sum_{\emptyset\ne S\subseteq[m]} x_{i,\tau,S} \le \typeprob_{i,\tau}
    && \forall i\in[n],\ \tau\in\Gamma_i, \label{eq:xos-lp-type}\\
& y_{i,j} = \sum_{\tau\in\Gamma_i} \sum_{S \ni j} x_{i,\tau,S}
    && \forall i\in[n],\ j\in[m], \label{eq:xos-lp-marginal}\\
& y_j=(y_{i,j})_{i\in[n]} \in \mathcal P_j
    && \forall j\in[m], \label{eq:xos-lp-item-polytope}\\
& x_{i,\tau,S} \ge 0
    && \forall i\in[n],\ \tau\in\Gamma_i,\ S\subseteq[m]. \notag
\end{align}
We denote the optimal value of \ref{lp:xos} by $\mathrm{LP\text{-}XOS}^*$.
Even with the more complicated feasibility constraints,
the analogous version of \Cref{clm: LP-OPT} is easily seen to hold:
\begin{fact}
\label{lem:xos-lp-upper-bound}
$\mathrm{LP\text{-}XOS}^*$ upper bounds the expected offline optimum.
\end{fact}
We are now ready to prove the reduction:
\begin{proof}[Proof of \Cref{thm:xos-ocrs-reduction}]

Let $x$ be an optimal solution to \ref{lp:xos},
and denote $\OCRS_j$ as the $\gamma_j$-selectable OCRS for $\mathcal{P}_j$.
When buyer $i$ arrives and reveals type $\tau$, sample a tentative bundle $T_i$ by
\[
    \Pr[T_i=A \mid v_i=v_{i,\tau}] = \frac{x_{i,\tau,A}}{\typeprob_{i,\tau}}
    \qquad \text{for every nonempty } A\subseteq[m],
\]
We declare buyer $i$ active for $\OCRS_j$, provided $j\in T_i$. Denoting this by the indicator random variable $Y_{i,j}$, observe that
by the sampling rule,
\begin{equation} \label{eqn:active_probability}
    \Pr[Y_{i,j}=1]
    = \sum_{\tau\in\Gamma_i} \typeprob_{i,\tau}
        \sum_{A\ni j} \frac{x_{i,\tau,A}}{\typeprob_{i,\tau}}
    = \sum_{\tau\in\Gamma_i}\sum_{A\ni j} x_{i,\tau,A}
    = y_{i,j}.
\end{equation}
Notice that \eqref{eqn:active_probability} together with \eqref{eq:xos-lp-item-polytope} implies that  $y_j\in\mathcal P_j$.
Moreover, the random variables $(Y_{i,j})_{i\in[n]}$ are independent across buyers, because the buyer types and tentative-bundle samples are independent across buyers.  

Our online algorithm is now simple to describe: When buyer $i$ arrives,
it receives exactly those items $j\in T_i$ for which  $\OCRS_j$ accepts buyer $i$.  Given this rule, the subset of items
assigned to agent $i$, denoted $\mathrm{ALG}_i$, is as follows:
\[
    \mathrm{ALG}_i := \{j\in T_i :  \OCRS_j \text{ accepts buyer }i\}.
\]
Since each $\OCRS_j$ selects a set of buyers in $\mathcal F_j$, the resulting allocation, $\cup_{i \in [n]}\mathrm{ALG}_i$ is feasible. 

We now analyze the competitive ratio of the online algorithm.
Fix $i\in[n]$. For any item $j\in A$, the event that buyer $i$ is assigned $j$ depends only on $(Y_{i,j})_{i \in [n]}$, and the internal randomness of $\OCRS_j$.  It does not depend on the exact type of buyer $i$, nor the exact instantiation of $T_i$.  Therefore, 
for any $\tau\in\Gamma_i$, and nonempty bundle $A\subseteq[m]$ with $x_{i,\tau,A}>0$, we may conclude that
\begin{equation}
\label{eq:xos-item-retention}
    \Pr[j\in \mathrm{ALG}_i \mid v_i=v_{i,\tau},\ T_i=A]
    = \Pr[\OCRS_j \text{ accepts buyer }i \mid Y_{i,j}=1]
    \ge \gamma_j,
\end{equation}
where the final line follows by \eqref{eq:selectability}, as $\OCRS_j$ was assumed
to be $\gamma_j$-selectable.

Since $\mathrm{ALG}_i\subseteq T_i=A$, the supporting clause \eqref{eq:xos-supporting-clause} gives
\[
    v_{i,\tau}(\mathrm{ALG}_i)
    \ge \sum_{j\in \mathrm{ALG}_i} a^{A}_{i,\tau,j}.
\]
Conditioning on $v_i=v_{i,\tau}$ and $T_i=A$, then using \eqref{eq:xos-item-retention}, we obtain
\[
    \mathbb E\bigl[v_{i,\tau}(\mathrm{ALG}_i) \mid v_i=v_{i,\tau},\ T_i=A\bigr]
    \ge \sum_{j\in A} a^{A}_{i,\tau,j}
        \Pr[j\in\mathrm{ALG}_i \mid v_i=v_{i,\tau},\ T_i=A]
    \ge \sum_{j\in A} \gamma_j a^{A}_{i,\tau,j}.
\]
Taking expectation over the type and tentative bundle of buyer $i$ gives
\[
\begin{aligned}
    \mathbb E[v_i(\mathrm{ALG}_i)]
    &\ge \sum_{\tau\in\Gamma_i} \typeprob_{i,\tau}
        \sum_{\emptyset\ne A\subseteq[m]}
        \frac{x_{i,\tau,A}}{\typeprob_{i,\tau}}
        \sum_{j\in A} \gamma_j a^{A}_{i,\tau,j} \\
    &= \sum_{\tau\in\Gamma_i}
        \sum_{\emptyset\ne A\subseteq[m]}
        x_{i,\tau,A}
        \sum_{j\in A} \gamma_j a^{A}_{i,\tau,j} .
\end{aligned}
\]
Summing over buyers and rearranging by items yields
\[
    \mathbb E\left[\sum_{i\in[n]} v_i(\mathrm{ALG}_i)\right]
    \ge
    \sum_{j\in[m]} \gamma_j
    \sum_{i\in[n]}\sum_{\tau\in\Gamma_i}\sum_{A\ni j}
        x_{i,\tau,A} a^{A}_{i,\tau,j}
    = \sum_{j\in[m]} \gamma_j W_j,
\]
where $W_j := \sum_{\tau\in\Gamma_i}\sum_{A\ni j}
        x_{i,\tau,A} a^{A}_{i,\tau,j}$.
Finally, by the supporting-clause identity at the tentative bundle,
\[
    \sum_{j\in[m]} W_j
    = \sum_{i\in[n]}\sum_{\tau\in\Gamma_i}\sum_{\emptyset\ne A\subseteq[m]}
        x_{i,\tau,A}\sum_{j\in A} a^{A}_{i,\tau,j}
    = \mathrm{LP\text{-}XOS}^*,
\]
where the last line follows since $x$ was assumed to be an optimal solution
to \ref{lp:xos}.
For $\gamma=\min_j\gamma_j$, the claimed bound on the online algorithm's performance follows, due to \Cref{lem:xos-lp-upper-bound}.
\end{proof}


%% file: Feige-appendix.tex
\section{Connection to Feige's algorithm}\label{sec:feige}

This section explains the relationship between the recursive BSG algorithm in \Cref{sec:recHeightBSG} and Feige's guiding graph algorithm~\cite{feige2006maximizing}. 
Our algorithm can be seen as a substantially simplified presentation of Feige's algorithm.
In this section, we informally describe how one can derive the recursive process by re-interpreting Feige's algorithm.
First, we will sketch Feige's algorithm and its analysis. After the sketch, we explain the conversion to our algorithm.

\textbf{Feige's Algorithm.}
Feige's guiding graph algorithm can be described briefly via the following steps. For a complete description, we refer the reader to the original text~\cite{feige2006maximizing}.
\begin{enumerate}
	\item \emph{Construct the guiding graph.} Fix any $n$-regular bipartite graph $G$ with large girth $g$. Edge-colour its edges into perfect matchings $M_1,\dots,M_n$ so that agent $i$ `controls' the edge set $M_i$.
	\item \emph{Assign edge labels.} For each $i$ and each edge $(u,v) \in M_i$, draw a label $S \subseteq [m]$ independently with $\Pr[label=S] = x_{i,S}$. Each edge $(u,v)$ can be thought of as a tentative set for agent $i$, distributed as per $(x_{i,S})_{S \subseteq [m]}$.
	\item \emph{Construct the item subgraphs.} For each item $j$, form the subgraph $G_j$ consisting only of the edges whose label $S$ contains $j$. Since the girth of $G$ is very large, one can show that with high probability, each subgraph $G_j$ is a forest.
	\item \emph{Orientation per item.} In every $G_j$, orient the edge of $G_j$ so that every vertex $u$ of $G_j$ has at most one incoming arc. This can always be done since $G_j$ is a forest. Note that the orientation of arcs may be different for each item.
	\item \emph{Tentative Sets.} Pick a vertex $u \in V(G)$ uniformly at random. For each agent $i$, let $u_i$ be the unique neighbour of $u$ in $M_i$. Assign to agent $i$ the label of the edge $(u,u_i)$ as his tentative set.
	\item \emph{Conflict Resolution.} An item $j$ could have conflicts. This means that $j$ is contained in both sets $S_i,S_j$, which are the labels of $(u,u_i),(u,u_j)$ respectively. In this scenario, we assign $j$ to the agent $i$ corresponding to the unique incoming arc $(u,v)$ for $u$. That is, agent $i$ is chosen so that the unique incoming arc $(u,v)$ is in $M_i$. If there is no such arc, we leave $j$ unallocated.  
\end{enumerate}

\textbf{Analysis.} We sketch the analysis of this algorithm below.

\begin{enumerate}
	\item \emph{Linearity at the random vertex.} For any fixed $u$, we have $\mathbb{E}\left[ \sum_{ i \in [n]} v_i(S(u,u_i)) \right]$ is equal to the objective value of the configuration LP solution $x$.
	\item \emph{Subadditivity gives $\frac{1}{2}$-factor per edge.} Fix a player $i$ and $S$ with $x_{i,S}>0$ and condition on the event that $S(u,u_i)=S$ and that every $G_j$ for $j \in S$, the connected component containing $u$ is a tree. Let $S_u$ (resp. $S_v$) denote the set of items assigned to player $i$ with the random vertex is chosen to be $u$ (resp. its other endpoint $v$ of $(u,u_i)$). By construction, $S \subseteq S_u \cup S_v$. Thus, $\mathbb{E}[v_i(S_i)] \geq \frac{1}{2}v_i(S)$, where $S_i$ is the final allocated set to agent $i$.
	\item \emph{Combining everything.} One can show that w.h.p. it holds that the connected component in $G_j$ containing $u$ is a tree for every $j \in [m]$. Hence, the final bound is $\mathbb{E}\left[ \sum v_i(S_i) \right] \geq(1-\epsilon) \frac{1}{2}LP$.
\end{enumerate}

To convert Feige's algorithm to our recursive bundle score generator algorithm, we make a couple of observations:
\begin{enumerate}
	\item First, certain explicit constructions of high girth regular graphs are also vertex transitive, such as the one in \cite{Lubotzky2017RamanujanG}. As a result, in Step 5, one does not need to sample a random vertex $u$ to define the allocation. One can simply fix a vertex $u$ of the graph beforehand deterministically. Due to the vertex transitivity of the graph, the local structure of the graph around a fixed vertex is identical to that of a randomly chosen vertex, and as such the resulting distribution of allocations (which uses only local information) will be identical whether we choose $u$ randomly or deterministically.
	\item The use of high-girth graphs suggests that really what we want is a $n$-regular graph that is acyclic. 
	The only such graph is the infinite $n$-regular tree.
	Since it is infinite, we cannot construct this graph explicitly. Instead, the tree is implicitly constructed as the recursion tree of our recursive algorithm. Indeed, every recursive call in our algorithm produces $n$ recursive subcalls. This corresponds to each node of its recursion tree having $n$ children, forming an $n$-ary tree.
	Furthermore, at every recursive call, the bundle score generator samples one interest set from the LP distribution of each player. This corresponds to the edges incident on each node having labels sampled from each player's distribution. 
	That is, sampling the labels of the edges incident to a node in the guiding graph is akin sampling these interest sets.
	Thus, the recursion tree of $\Height$ can be viewed as the tree generated by a depth-first-search on the guiding graph.
	\item One particular choice of orientations in Step 4 is to orient every edge away from the diametric center of the tree. A vertex $w$ is said to be the diametric center if it minimizes the maximum distance to any other vertex in the tree. For every edge $uv$ of the tree $G_j$, we orient $uv$ away from the component of $G_j-\{uv\}$ that contains the diametric center $w$. 
	At the root $u$, this would mean that the unique incoming arc to $u$ comes from the vertex $u_i$ for which the subtree of $G_j$ rooted at the vertex $u_i$ has the $u_i$-to-leaf path of largest length. In other words, an item $j$ is allocated to the player $i$ for which the subtree of $G_j$ rooted at the vertex $u_i$ has the largest height. Correspondingly, our algorithm computes this height recursively via the $1+\max y_{ij}$ recurrence (hence giving {\Height} its name). 
\end{enumerate}


%% file: Odd-n-Appendix.tex
\section{Improved Value Retention for IID Arrivals for Odd $n$}
\label{sec: iid-odd-n}

Here, we give the complete proof of the adaptation of \Cref{lem:one-fifth-val} for the case of odd $n$ -- the improved ratio of $60/11$ then follows immediately as before. 

\begin{lemma}
\label{lem:one-fifth-val-odd}
    Suppose that $n=2h+1$ is odd. For \Cref{alg:online-algorithm} in the
    i.i.d. arrival model, fix any buyer $i\le h+1$. Then, for every type
    $\type\in\typeset$ and every tentative bundle $A\subseteq[m]$,
    \[
        \E\!\left[
            v_i(\ALG_i)
            \,\middle|\,
            v_i=v_{\type},\ T_i=A
        \right]
        \ge
        \frac{1}{5}v_{\type}(A).
    \]
    The expectation is over all remaining randomness of the algorithm, including
    the types and scores of earlier buyers and the threshold scores.
\end{lemma}

\begin{proof}
    We use the notation from the proof of \Cref{lem:one-sixth-val}. Condition on
    $v_i=v_{\type}$ and $T_i=A$. Then
    \[
        \ALG_i
        =
        \{j\in A: b_{i,j}>p_j'\ge b^{<i}_j\}.
    \]

    We complete the actual past to $h$ buyer-score vectors. For each
    $k\in\{i,\dots,h\}$, independently draw
    $T_k^{\mathrm{mid}}\sim\mu$, set
    $b_k^{\mathrm{mid}}=\Height(T_k^{\mathrm{mid}})$, and define
    \[
        p_j^{\mathrm{mid}}
        :=
        \max\Bigl(
            \{b^{<i}_j\}
            \cup
            \{b_{k,j}^{\mathrm{mid}}:
                i\le k\le h,\ j\in T_k^{\mathrm{mid}}\}
        \Bigr)
        \qquad \forall j\in[m].
    \]
    When $i=h+1$, the actual past
    consists of exactly $h$ buyers. Since
    $p_j^{\mathrm{mid}}\ge b^{<i}_j$ for every $j$, the set
    \[
        B^{\mathrm{mid}}
        :=
        \{j\in A: b_{i,j}>p_j'\ge p_j^{\mathrm{mid}}\}
    \]
    satisfies
    \[
        B^{\mathrm{mid}}
        \subseteq
        \{j\in A: b_{i,j}>p_j'\ge b^{<i}_j\}
        =
        \ALG_i.
    \]

    Let $\mathcal{D}$ be the distribution, restricted to $A$, of one
    buyer-score vector generated from a tentative bundle sampled according to
    $\mu$. Thus, to draw from $\mathcal{D}$, we draw $T\sim\mu$, generate
    $\Height(T)$, and give every coordinate in $A\setminus T$ score zero.
    Let $\mathcal{F}_h$ be the distribution of the coordinatewise maximum of
    $h$ independent draws from $\mathcal{D}$. When $h=0$,
    $\mathcal{F}_h$ is the point mass at the all-zero vector.

    Draw $W_1,W_2$ independently from $\mathcal{F}_h$, and independently draw
    $a_1$ from $\mathcal{D}$. Since $\max\{W_1,W_2,a_1\}$
    is the coordinatewise maximum of $2h+1=n$ independent buyer-score vectors,
    it is distributed as a full threshold vector restricted to $A$.
    Therefore, \Cref{thm:bsg-guarantees} gives
    \[
        b_i=\Height(A)
        \succ
        \max\{W_1,W_2,a_1\}.
    \]
    Hence there is a coupling $\omega$ of $b_i$ and
    $\max\{W_1,W_2,a_1\}$ such that
    \[
        b_{i,j}>
        \max\{W_{1j},W_{2j},a_{1j}\}
        \qquad\text{for every }j\in A
    \]
    simultaneously.

    Next draw $W_3,W_4,W_5$ independently from $\mathcal{F}_h$, and draw
    $a_2$ independently from $\mathcal{D}$. These vectors are also independent
    of the coupled pair. Since the buyers are i.i.d.,
    \[
        p'|_A
        \ \overset{d}{=}\
        \max\{W_3,W_4,a_2\},
    \]
    while, by the construction of $p^{\mathrm{mid}}$,
    \[
        p^{\mathrm{mid}}|_A
        \ \overset{d}{=}\
        W_5.
    \]
    Moreover, conditional on $v_i=v_{\type}$ and $T_i=A$, the vectors
    $b_i$, $p'$, and $p^{\mathrm{mid}}$ are mutually independent. Define
    \[
        B_W
        :=
        \left\{
            j\in A:
            b_{i,j}>
            \max\{W_{3j},W_{4j},a_{2j}\}
            \ge
            W_{5j}
        \right\}.
    \]
    Using the monotonicity of $v_{\type}$ and
    $\ALG_i\supseteq B^{\mathrm{mid}}$, we have
    \begin{align}
    \label{eq:odd-exp-eq}
        \E\!\left[
            v_i(\ALG_i)
            \,\middle|\,
            v_i=v_{\type},\ T_i=A
        \right]
        &\ge
        \E[v_{\type}(B^{\mathrm{mid}})]
        =
        \E[v_{\type}(B_W)],
    \end{align}
    where the rightmost expectation is taken under the coupling $\omega$ and
    over the independent draws $W_3,W_4,W_5,a_2$.

    Define
    \[
        \widetilde B
        :=
        \left\{
            j\in A:
            \max\{W_{1j},W_{2j},a_{1j}\}
            \ge
            \max\{W_{3j},W_{4j},a_{2j}\}
            \ge
            W_{5j}
        \right\}.
    \]
    Under the coupling $\omega$, if $j\in\widetilde B$, then
    \[
        b_{i,j}>
        \max\{W_{1j},W_{2j},a_{1j}\}
        \ge
        \max\{W_{3j},W_{4j},a_{2j}\}
        \ge
        W_{5j}.
    \]
    Hence $\widetilde B\subseteq B_W$. Using the monotonicity of
    $v_{\type}$ and Eq.~\eqref{eq:odd-exp-eq}, we obtain
    \begin{align}
    \label{eq:odd-b-tilde-lb}
        \E\!\left[
            v_i(\ALG_i)
            \,\middle|\,
            v_i=v_{\type},\ T_i=A
        \right]
        \ge
        \E[v_{\type}(B_W)]
        \ge
        \E[v_{\type}(\widetilde B)],
    \end{align}
    where the rightmost expectation is taken over independent draws
    $W_1,\dots,W_5,a_1,a_2$.

    It remains to lower bound this final expectation. Label the vertices of a
    $5$-cycle in the order $1-5-2-3-4-1$. For $l\in[5]$, let $N(l)$ denote
    the two neighbors of $l$ in this cycle. Define
    \[
        (c_l,d_l)
        :=
        \begin{cases}
            (a_2,a_1), & l\in\{1,2\},\\
            (a_1,a_2), & l\in\{3,4,5\}.
        \end{cases}
    \]
    For each $l\in[5]$, define
    \[
        R_l
        :=
        \left\{
            j\in A:
            \max\left(
                \{W_{rj}:r\in N(l)\}
                \cup
                \{c_{l,j}\}
            \right)
            \ge
            \max\left(
                \{W_{rj}:r\notin N(l),\ r\neq l\}
                \cup
                \{d_{l,j}\}
            \right)
            \ge
            W_{lj}
        \right\}.
    \]
    With the chosen cyclic order and the definition of $(c_5,d_5)$, the set
    $\widetilde B$ is exactly $R_5$.

    We claim that the sets $R_1,\dots,R_5$ cover $A$ for every realization of
    $W_1,\dots,W_5,a_1,a_2$. Fix $j\in A$, and consider a largest value among
    \[
        W_{1j},\dots,W_{5j},a_{1j},a_{2j}.
    \]

    First, suppose that $W_{rj}$ is a largest value, and let $u,w$ be the two
    neighbors of $r$ in the cycle. Since
    $r\in N(u)\cap N(w)$, the first inequality in the definitions of both
    $R_u$ and $R_w$ holds. Moreover, $u$ and $w$ are nonadjacent. If
    $j\notin R_u$, then the second inequality for $R_u$ fails, so
    $W_{uj}>W_{wj}$. If $j\notin R_w$, then similarly
    $W_{wj}>W_{uj}$. These two inequalities cannot both hold. Hence
    \[
        j\in R_u\cup R_w.
    \]

    Next, suppose that $a_{1j}$ is a largest value. Since
    $c_3=c_5=a_1$, the first inequality holds for both $R_3$ and $R_5$.
    The vertices $3$ and $5$ are nonadjacent. If $j\notin R_3$, then the
    second inequality for $R_3$ fails, so $W_{3j}>W_{5j}$. If
    $j\notin R_5$, then similarly $W_{5j}>W_{3j}$. Hence
    \[
        j\in R_3\cup R_5.
    \]

    Finally, suppose that $a_{2j}$ is a largest value. Since
    $c_1=c_2=a_2$, the first inequality holds for both $R_1$ and $R_2$.
    The vertices $1$ and $2$ are nonadjacent. If $j\notin R_1$, then
    $W_{1j}>W_{2j}$, while if $j\notin R_2$, then
    $W_{2j}>W_{1j}$. Hence
    \[
        j\in R_1\cup R_2.
    \]
    Ties cause no difficulty because all the defining inequalities are weak.
    This proves the cover claim.

    By subadditivity,
    \[
        \sum_{l=1}^5 v_{\type}(R_l)\ge v_{\type}(A).
    \]
    Since $W_1,\dots,W_5$ are i.i.d., and $a_1,a_2$ are i.i.d. and
    independent of the $W$ vectors, the random sets
    $R_1,\dots,R_5$ are identically distributed. Indeed, the $W$ vectors may
    be permuted and $a_1,a_2$ may be interchanged without changing their joint
    distribution. The sets need not be independent. Therefore, since
    $\widetilde B=R_5$, by \eqref{eq:odd-b-tilde-lb},
    \[
        \E\!\left[
            v_i(\ALG_i)
            \,\middle|\,
            v_i=v_{\type},\ T_i=A
        \right]
        \ge
        \E[v_{\type}(R_5)]
        =
        \frac{1}{5}
        \E\!\left[\sum_{l=1}^5v_{\type}(R_l)\right]
        \ge
        \frac{1}{5}v_{\type}(A).
    \]
\end{proof}